\documentclass[sn-basic, Numbered]{sn-jnl}

\usepackage{graphicx}%
\usepackage{multirow}%
\usepackage{amsmath,amssymb,amsfonts}%
\usepackage{amsthm}%
\usepackage{mathrsfs}%
\usepackage[title]{appendix}%
\usepackage{xcolor}%
\usepackage{textcomp}%
\usepackage{manyfoot}%
\usepackage{booktabs}%
\usepackage{algorithm}%
\usepackage{algorithmicx}%
\usepackage{algpseudocode}%
\usepackage{listings}%

\usepackage{booktabs} 
\usepackage{float}
\usepackage{graphicx} 
\usepackage{natbib}

\usepackage{graphicx}
\usepackage{subcaption}  

\usepackage{epstopdf}
\usepackage{ulem}

\theoremstyle{thmstyleone}%
\newtheorem{theorem}{Theorem}
\theoremstyle{thmstyletwo}%
\newtheorem{remark}{Remark}%

\theoremstyle{thmstylethree}%
\newtheorem{definition}{Definition}%

\newtheorem{Cor}{Corollary}

\begin{document}

\title[Variation of entropy in the  Duffing system with the amplitude of the external force]{Variation of entropy in the  Duffing system with the amplitude of the external force}


\author[1]{\fnm{Junfeng} \sur{Cheng}}\email{chengjf@hust.edu.cn}

\author*[1,2]{\fnm{Xiao-Song} \sur{Yang}}\email{yangxs@hust.edu.cn}

\affil*[1]{\orgdiv{School of Mathematics}, \orgname{Huazhong University of Science and Technology}, \orgaddress{\city{Wuhan}, \postcode{430074}}}

\affil[2]{\orgname{Hubei Key Laboratory of Engineering Modeling and Scientific Computing}, \orgaddress{ \city{Wuhan}, \postcode{430074}}}


\abstract{In this paper, we revisit the well-known perturbed Duffing system and investigate its chaotic dynamics by means of numerical Runge--Kutta method based on topological horseshoe theory.
	Precisely, we investigate chaos through the topological horseshoes associated with the first, second, and third return maps, obtained by varying the amplitude of an external force term while keeping all other parameters fixed.
	Our new finding demonstrates that, when the force amplitude exceeds a certain value, the topological (Smale) horseshoe degenerates into a pseudo-horseshoe, while chaotic invariant set persists. This phenomenon indicates that the lower bound of the topological entropy decreases as the force amplitude increases, thereby enriching the dynamics in the perturbed Duffing system.
	
	Furthermore, we identify a critical value of the force amplitude governing the attractivity of the chaotic invariant set. For amplitudes slightly below this value, the basin of attraction of the chaotic invariant set progressively shrinks as the amplitude increases. In contrast, for larger amplitudes, both Lyapunov exponents become negative while the topological horseshoe persists, suggesting that the chaotic invariant set loses attractivity as the amplitude grows.}




\maketitle

\section{Introduction}\label{sec1}

The study of nonlinear dynamical systems has attracted considerable attention across various scientific and engineering disciplines. Among these systems, the Duffing oscillator stands out as a typical model used to investigate nonlinear behaviors. Originally introduced by Georg Duffing in the early 20th century, the Duffing system has since been instrumental in exploring complex phenomena, including bifurcations and chaos.

Classical results on the Duffing system can be found in the books of J.K. Hale \cite{Hale1980}, Guckenheimer and Holmes \cite{Guckenheimer1984}. In Ref.~\cite{Guckenheimer1984}, the Melnikov method,  together with the Smale-Birkhoff theorem \cite{Smale1965,Smale1967}, has confirmed that under certain conditions, the perturbed Duffing system exhibits chaotic dynamics,  as outlined in Section \ref{DSMM}. 
Recent investigation of the sinusoidally forced Duffing oscillator \cite{Abohamer2025} employed perturbation techniques in conjunction with bifurcation diagrams, Poincar\'e sections, and Lyapunov exponents to characterize transitions to chaos, thereby producing detailed parameter maps and insights of relevance to engineering applications.
In Ref.~\cite{Cheng2025b}, it was further verified, via the topological horseshoe theory, that the Melnikov criterion is sufficient but not necessary for the existence of chaotic dynamics.
In this paper, however, we attempt to use similar methodology to provide a more detailed description of the influence brought by the external force term.

The exploration into topological horseshoe theory was initiated by the proposal of the concept of the Smale horseshoe \cite{Smale1967}.  After that, the notion of the topological horseshoe was introduced by Kennedy and Yorke\citep{Kennedy2001-1,Kennedy2001-2}. Subsequently, Yang and Tang enhanced this concept in Ref.~\cite{Yang2004} to deal with piecewise continuous maps. With this theorem, we can use numerical methods, such as the Runge-Kutta method, to study chaos effectively.

Furthermore, we investigate how the attractivity of the chaotic invariant set varies with the external force in the Duffing system. Specifically, we utilize Lyapunov exponents, a powerful tool for detecting and quantifying chaotic behavior, to demonstrate that as the amplitude of the external force  increases, notable changes in the attractivity of the chaotic set emerge.
In addition, we propose a classification procedure inspired by Cauchy's convergence criterion to further illustrate the intricate geometric structure of the corresponding attraction basin.

The paper is organized as follows. 	Sec.~\ref{DSMM} reviews the Duffing system together with a classical result obtained via the Melnikov method. Sec.~\ref{THT} provides an overview of the topological horseshoe theory. 
Sec.~\ref{experiment1} presents our analysis of the Smale horseshoes and pseudo-horseshoes associated with the first, second, and third return maps, showing that the lower bound of the topological entropy decreases as the amplitude of the external force increases. 
Sec.~\ref{experiment2} investigates how external force affects the attractivity of the chaotic invariant set, identifying a critical threshold.
When the force amplitude is slightly below this threshold, the basin of attraction of the chaotic invariant set gradually decreases as the amplitude increases. Once the amplitude exceeds the threshold, however, the chaotic invariant set ceases to be attractive, and a stable harmonic solution appears.
Finally, Sec.~\ref{conclusion} summarizes the main findings of our study.

\section{Preliminaries}\label{preliminaries}

\subsection{Duffing system and a classical result obtained by the Melnikov method}\label{DSMM}

The Duffing system is governed by the following second-order nonlinear differential equation:
\begin{eqnarray}\label{Duffing-1}
	\ddot{x} + \delta \dot{x} + \alpha x + \beta x^3 = \gamma \cos(\omega t),
\end{eqnarray}  
where \( x(t) \) represents displacement, \( \delta \) is the damping coefficient, \( \alpha \) and \( \beta \) represent linear and nonlinear stiffness coefficients, respectively, \( \gamma \) denotes the amplitude of the external driving force, and \( \omega \) is the angular frequency of the applied force.

For the purposes of this paper, we fix \( \alpha = -1 \) and \( \beta = 1 \), transforming \eqref{Duffing-1} into:
\begin{eqnarray}\label{Duffing-2}
	\ddot{x} + \delta \dot{x} - x +  x^3 = \gamma \cos(\omega t).
\end{eqnarray} 
By introducing \( y = \dot{x} \) and a phase variable \( \theta \), we can reformulate the system into a autonomous system:
\begin{eqnarray}\label{Duffing3}
	\left\{\begin{array}{l}
		\dot{x}=y,\\
		\dot{y}=x-x^3-\delta y+\gamma \cos(\theta),\qquad(x,y,\theta)\in\mathbb{R}^2\times S^1,\\
		\dot{\theta}=\omega.
	\end{array}\right.
\end{eqnarray}

By introducing $\varepsilon$, and letting $\delta=\varepsilon\hat{\delta}, \gamma=\varepsilon\hat{\gamma}$, we have the classical result obtained by the Melnikov method, as presented in the book by Wiggins \cite{Wiggins2013}. To be specific, consider the following system.
\begin{eqnarray}\label{Duffing3.1}
	\left\{\begin{array}{l}
		\dot{x}=y,\\
		\dot{y}=x-x^3+\varepsilon (-\hat{\delta} y+ \hat{\gamma} \cos(\theta)),\qquad(x,y,\theta)\in\mathbb{R}^2\times S^1,\\
		\dot{\theta}=\omega.
	\end{array}\right.
\end{eqnarray}

As noted in Ref.~\cite{Wiggins2013}, a sufficient condition for the stable and unstable manifolds to intersect transversely is
\begin{eqnarray}\label{MMc}
	\hat\delta  < \left( \frac{{3\pi \omega }\operatorname{sech}\frac{\pi \omega }{2}}{2\sqrt{2}}\right) \hat\gamma.  
\end{eqnarray}
Since
$\delta=\varepsilon\hat{\delta}, \gamma=\varepsilon\hat{\gamma}$, condition \eqref{MMc} is equivalent to 
\begin{eqnarray}\label{MMcc}
	\delta  < \left( \frac{{3\pi \omega }\operatorname{sech}\frac{\pi \omega }{2}}{2\sqrt{2}}\right) \gamma.  
\end{eqnarray}
By the Smale--Birkhoff theorem \cite{Guckenheimer1984}, when \eqref{MMcc} holds,  some iterate of the Poincar\'e map  exhibits a Smale horseshoe. Consequently, a topological horseshoe exists, implying the presence of chaotic dynamics in the system.

\subsection{Topological Horseshoe}\label{THT}

In this section, we present a review of topological horseshoes for the purposes of this paper.
Prior to exploring the concept of topological horseshoes, it is crucial to introduce the definition of semi-conjugation. The subsequent two definitions are referenced from Ref.~\cite{Wiggins1988}.

\begin{definition} \cite{Wiggins1988}
	Let $M$ and $N$ be topological spaces and consider two continuous maps $f:M\to M$ and $g:N\to N$. The map $f$ is said to be semi-conjugate to $g$ if there is a continuous surjective map $h:M\to N $  such that 
	\begin{eqnarray}
		h\circ f=g\circ h.
	\end{eqnarray}
\end{definition}

\begin{definition} \cite{Wiggins1988}
	Let $X$ be a metric space, and let $\mathcal{A}=\{0,1,2,\dots,m-1\}$ be a finite alphabet with $m$ symbols. 
	The $m$-shift map $\sigma$ is defined as 
	\begin{eqnarray}
		\sigma(s)_i=s_{i+1},
	\end{eqnarray}
	where $s\in\mathcal{A}^\mathbb{Z}=\{x=(x_i)_{i\in\mathbb{Z}}:x_i\in\mathcal{A}~\text{for all } i\in\mathbb{Z}\}$.
	
	Consider a  (piecewise) continuous map $f:X\to X$. If there exists a compact invariant set $\Lambda \subset X $ such that the restriction of $f$ to $\Lambda$ is semi-conjugate to the m-shift map $\sigma $,
	then $f$ is said to have an m-type topological horseshoe.  
\end{definition}

For the practical identification of horseshoes within applied problems, we also need to introduce the concept of crossing, detailed in Ref.~\cite{Yang2009}. Let us consider a compact and connected region \(D \subset \mathbb{R}^n\). Let \(B_i\) (\(i = 1, 2, \ldots, m\)) be compact, path-connected subsets  in \(D\), each homeomorphic to the unit cube. Denote the boundary of each set \(B_i\) by \(\partial B_i\), and consider a piecewise continuous map $f: D \rightarrow X$, which is continuous on each of the compact sets $B_i$.

\begin{definition} \citep{Yang2009}
	For each $B_i, 1\leq i\leq m$, let $B_i^1$ and $B_i^2$ be two fixed disjoint connected nonempty compact subsets (usually pieces of $\partial B_i$) contained in the boundary $\partial B_i$. A connected subset $l$ of $B_i$ is said to be a connection of $B_i^1$ and $B_i^2$ if $l\cap B_i^1\neq \emptyset$ and $l\cap B_i^2\neq \emptyset$. 
\end{definition}

\begin{definition}\label{thd} \citep{Yang2009}
	Let $l\subset B_i$ be a connection of $B_i^1$ and $B_i^2$. We say that $f(l)$ is crossing $B_j$, if $l$ contains a connected subset $\bar{l}$ such that $f(\bar{l})$ is a connection of $B_j^1$ and $B_j^2$, i.e., $f(\bar{l})\subset B_j$, while $f(\bar{l})\cap B_j^1\neq \emptyset$ and $f(\bar{l})\cap B_j^2\neq \emptyset$. In this case we denote it by $f(l)\mapsto B_j$. Furthermore, if $f(l) \mapsto B_j $ for every connection $l$ of $B_i^1$ and $B_i^2$, then $f(B_i)$ is said to be crossing $B_j$ and denoted by $f(B_i)\mapsto B_j$. To simplify the terminology, we refer to $B_i$ as a crossing block of $B_j$, and for clarity, we call this crossing the dimension one crossing.
\end{definition}

In light of the aforementioned definitions, we recall the following theorem.
\begin{theorem}\label{thm1} \citep{Yang2004} 
	Suppose that the map $f:D\to \mathbb{R}^n$ satisfies the following assumptions:\\		
	\noindent(1) There exist $m$ mutually path-connected disjoint compact subsets $B_1,B_2,\dots$ and $B_m$ of $D$, the restriction of $f$ to each $B_i$, i.e., $f|_{B_i}$ is continuous.\\		
	\noindent (2) The dimension one crossing relation $f(B_i)\mapsto B_j $ holds for $1\leq i,j\leq m$.\\	
	Then there exists a compact invariant set $K\subset D$, such that $f|_K$ is semi-conjugate to a m-shift map.
\end{theorem}

The Smale horseshoe mentioned in Section \ref{DSMM} is a special case of topological horseshoes.
It is obvious that the existence of a Smale horseshoe results in the existence of a topological horseshoe consisting of two crossing blocks.

To describe the chaos of a system quantitatively, we recall a theorem on the topological entropy. For more details, readers are referred to Ref.~\cite{Robinson1995}.

\begin{theorem}\label{TE} \citep{Robinson1995}
	Let $X$ be a compact space, and $f:X\to X$ a continuous map. If there exists an invariant set $\Lambda\subset X$ such that $f|_\Lambda$ is semi-conjugate to the m-shift map $\sigma$, then we have
	\begin{eqnarray}
		h(f)\geq h(\sigma)=\log m,
	\end{eqnarray}
	where $h(f)$ denotes the topological entropy of the map $f$. Furthermore, for every positive integer $k$, we have the following fact
	\begin{eqnarray}
		h(f^k)=kh(f).
	\end{eqnarray}
\end{theorem}

\section{Variation of entropy in the Duffing system in terms of geometrical structure of horseshoe}\label{experiment1}

It has been verified in Ref.~\cite{Cheng2025b} that the chaotic invariant set may persist even in the absence of the Melnikov criterion. In this paper, however, we investigate how the amplitude of the external force influences the dynamics, employing topological horseshoe theory.
To be specific, we focus on the case with the parameters specified in Ref.~\cite{Guckenheimer1984}:
\begin{eqnarray}
	\omega=1,\delta=0.25.
\end{eqnarray}
Substituting them into \eqref{MMcc}, the constant on the right-hand side of \eqref{MMcc} satisfies
\begin{eqnarray}
	\frac{3\pi \text{sech}(\frac{\pi}{2})}{2\sqrt{2}}\approx 1.327~99.
\end{eqnarray}
Consequently, the Melnikov condition \eqref{MMcc} holds whenever $\gamma > 0.188~254~527. $

Following Ref.~\cite{Guckenheimer1984}, we first consider the case when $\gamma=0.4$. A Smale horseshoe is presented in Figure \ref{fig1:a}, where the quadrilateral \( Q \) is selected with endpoints defined as:
\[Q: A(-0.23,0.05), ~B(0.13,0.40), ~C(0.13,0.20), ~D(-0.23,-0.15).\]
In the figures appearing throughout the remainder of this paper, each side of the quadrilateral is mapped to the side indicated by the corresponding color.
For readers' convenience, we present a schematic diagram for the Smale horseshoe for $\gamma=0.4$ in Figure~\ref{fig1:b}.

\begin{figure}[H]
	\centering
	\includegraphics[width=\textwidth]{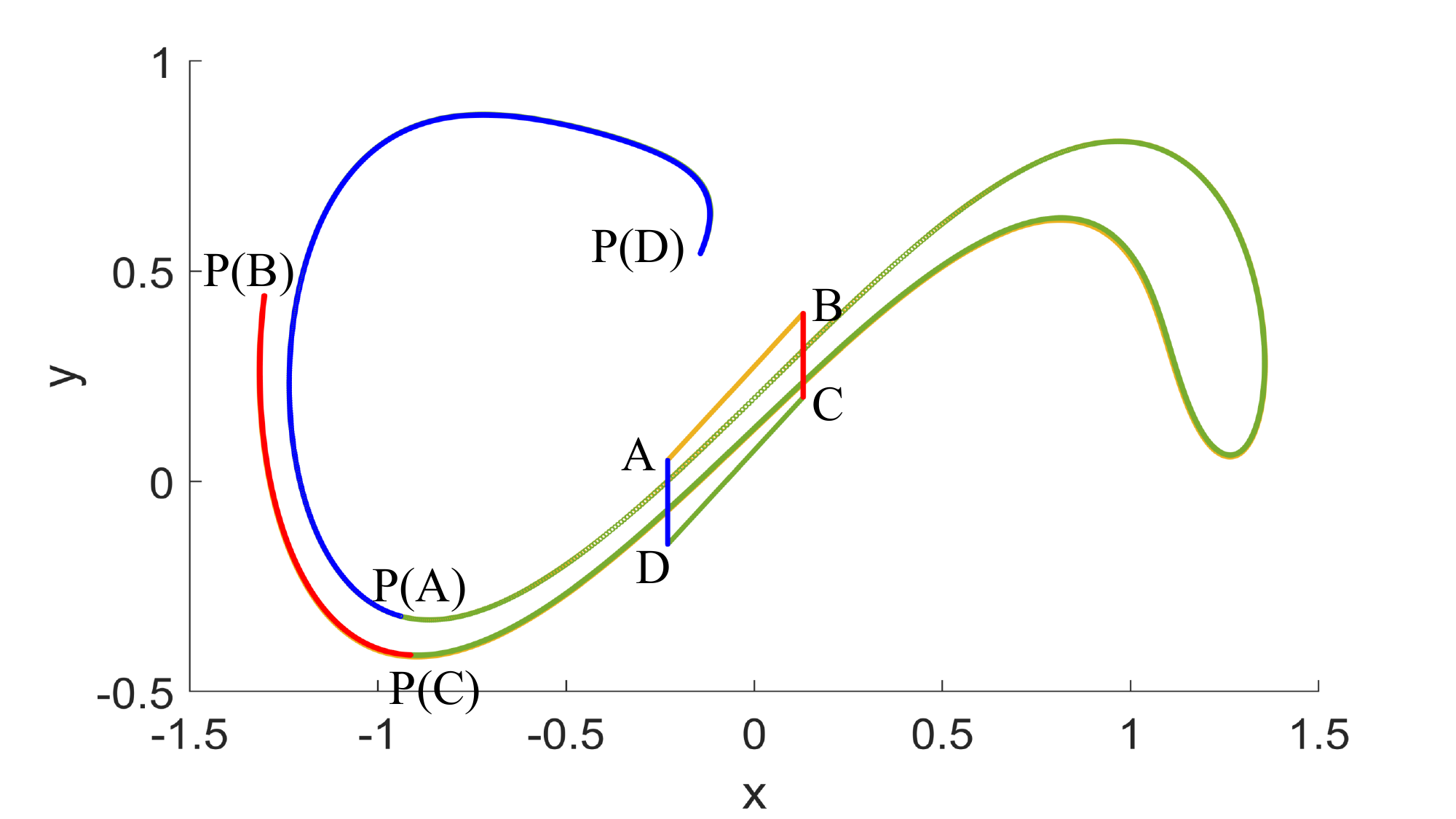}
	\caption{Smale horseshoe of the first return map when $\gamma=0.4$.}
	\label{fig1:a}
\end{figure}

\begin{figure}[H]
	\centering
	\includegraphics[width=\textwidth]{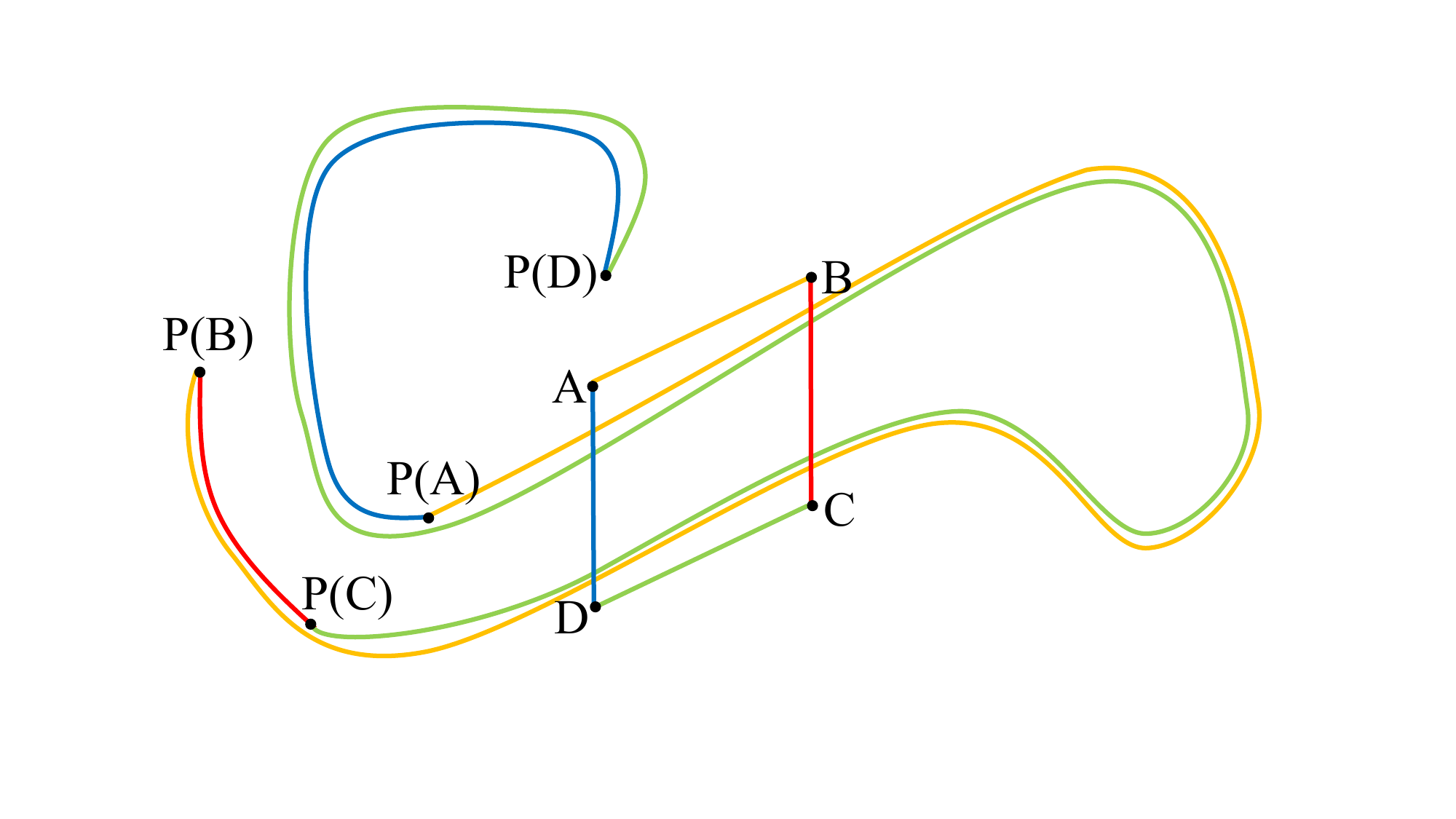}
	\caption{Schematic diagram of the Smale horseshoe in Figure~\ref{fig1:a}.}
	\label{fig1:b}
\end{figure}

Since the Melnikov condition is satisfied when the amplitude of the external force becomes sufficiently large, the existence of a chaotic invariant set is then guaranteed. 
Our interest, however, lies in understanding how the associated horseshoe structure changes as $\gamma$ increases.
To this end, we conduct a series of numerical experiments to identify the Smale horseshoe of the first return map, and several representative examples are presented in Figure~\ref{fig2:a}--\ref{fig2:h}.
\begin{figure}[H]
	\centering
	\includegraphics[width=\textwidth]{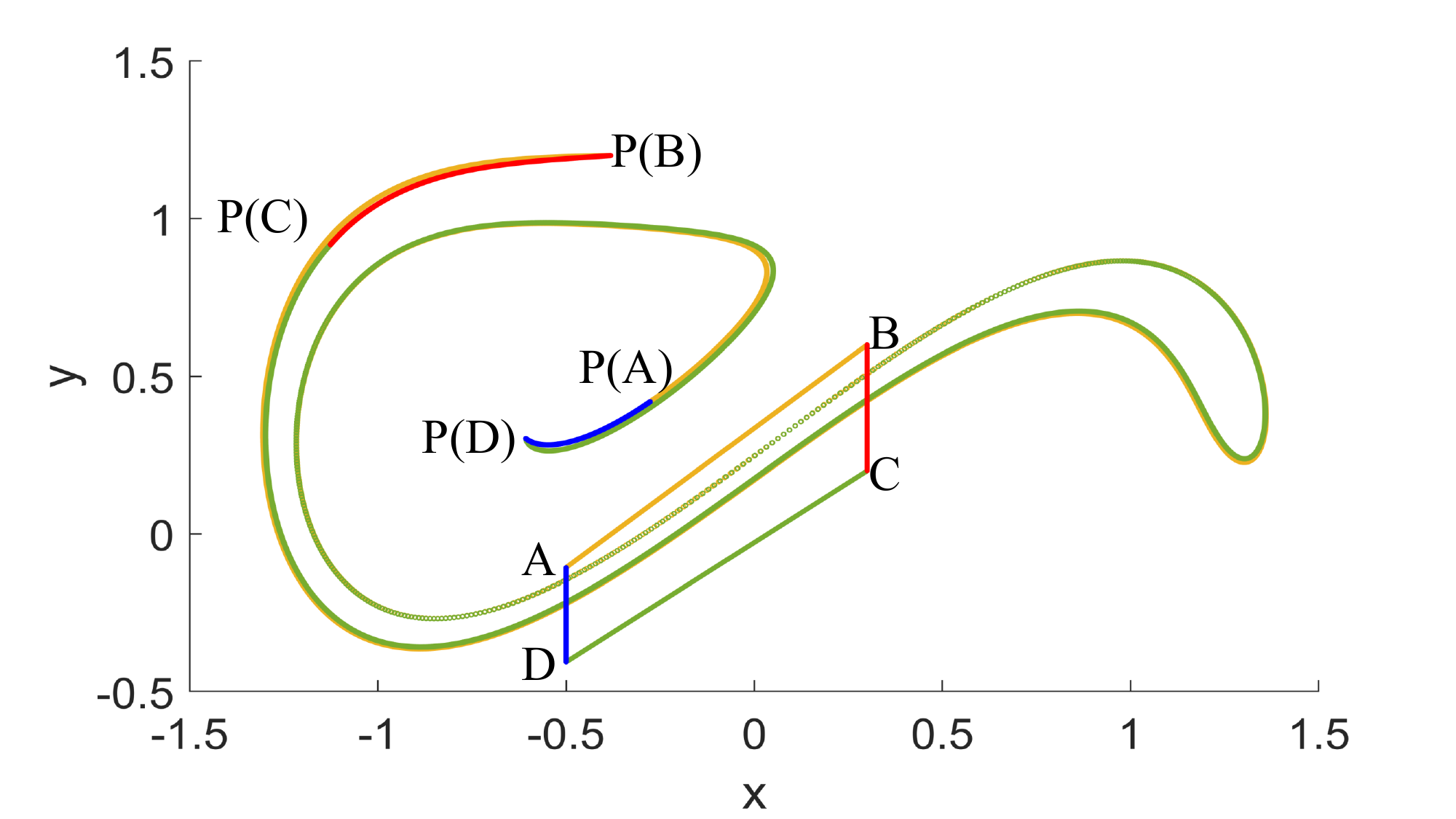}
	\caption{Smale horseshoe of the first return map when $\gamma=0.5$.}
	\label{fig2:a}
\end{figure}
\begin{figure}[H]
	\centering
	\includegraphics[width=\textwidth]{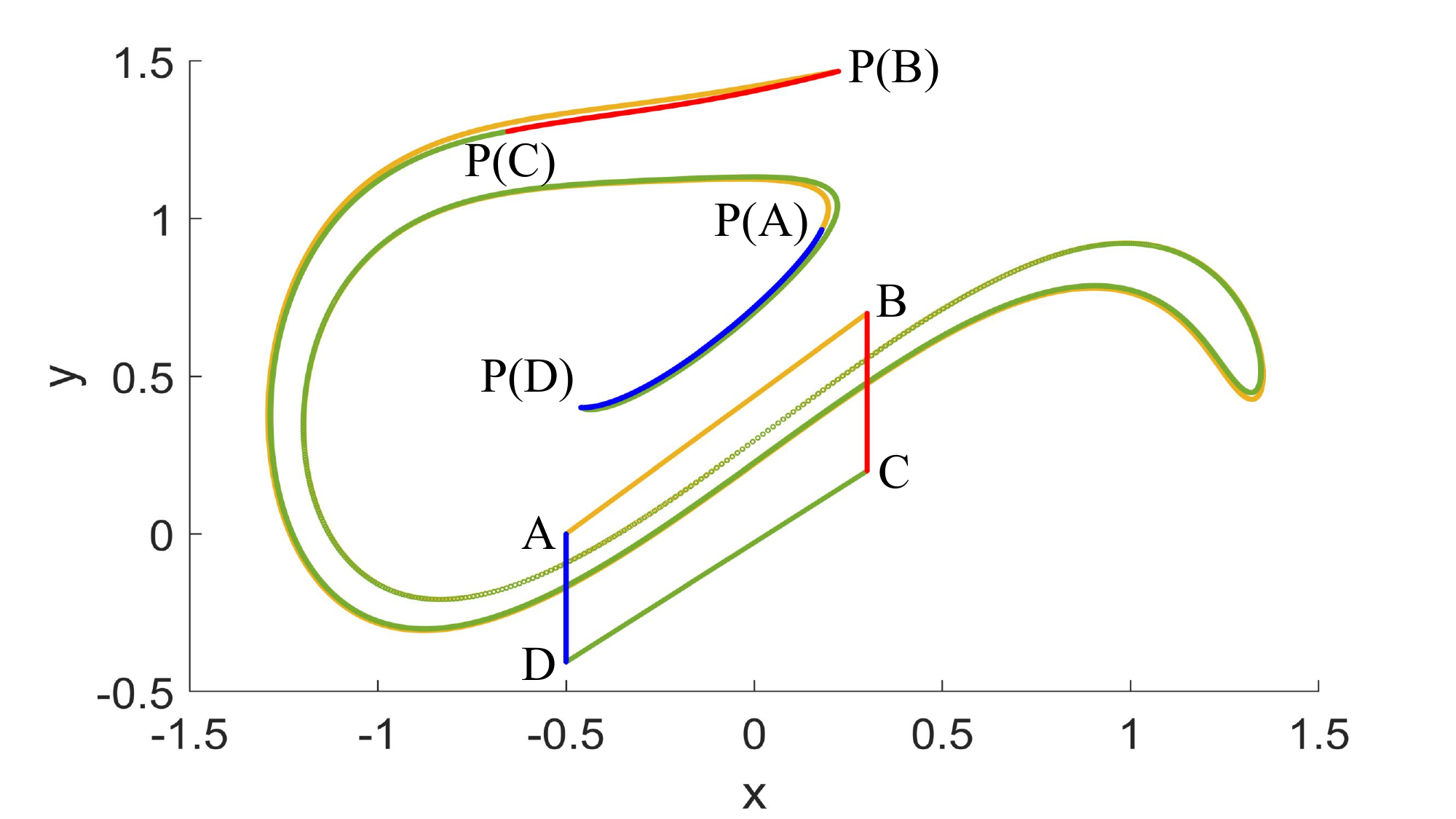}
	\caption{Smale horseshoe of the first return map when $\gamma=0.6$.}
	\label{fig2:b}
\end{figure}
\begin{figure}[H]
	\centering
	\includegraphics[width=\textwidth]{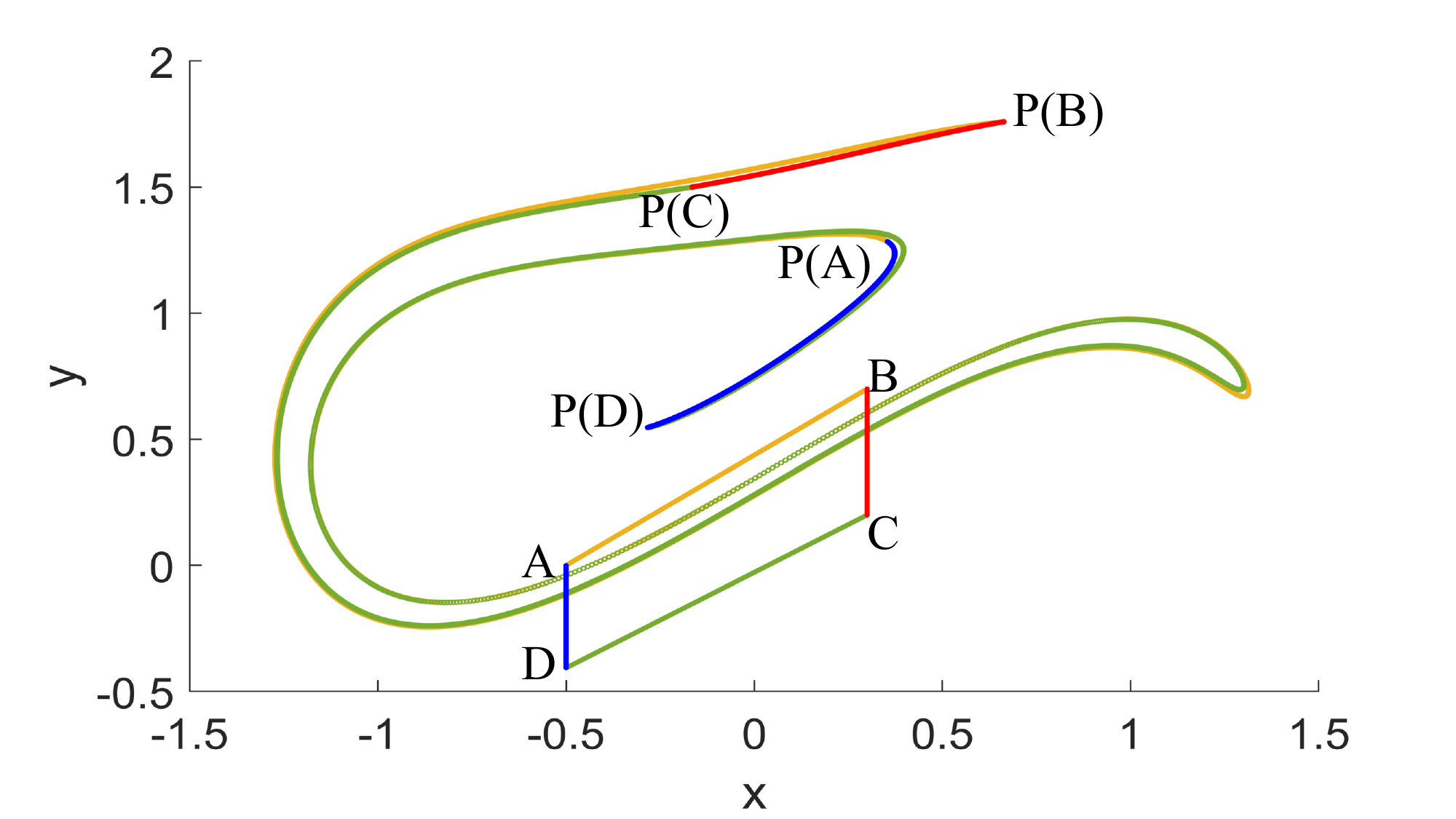}
	\caption{Smale horseshoe of the first return map when $\gamma=0.7$.}
	\label{fig2:c}
\end{figure}
\begin{figure}[H]
	\centering
	\includegraphics[width=\textwidth]{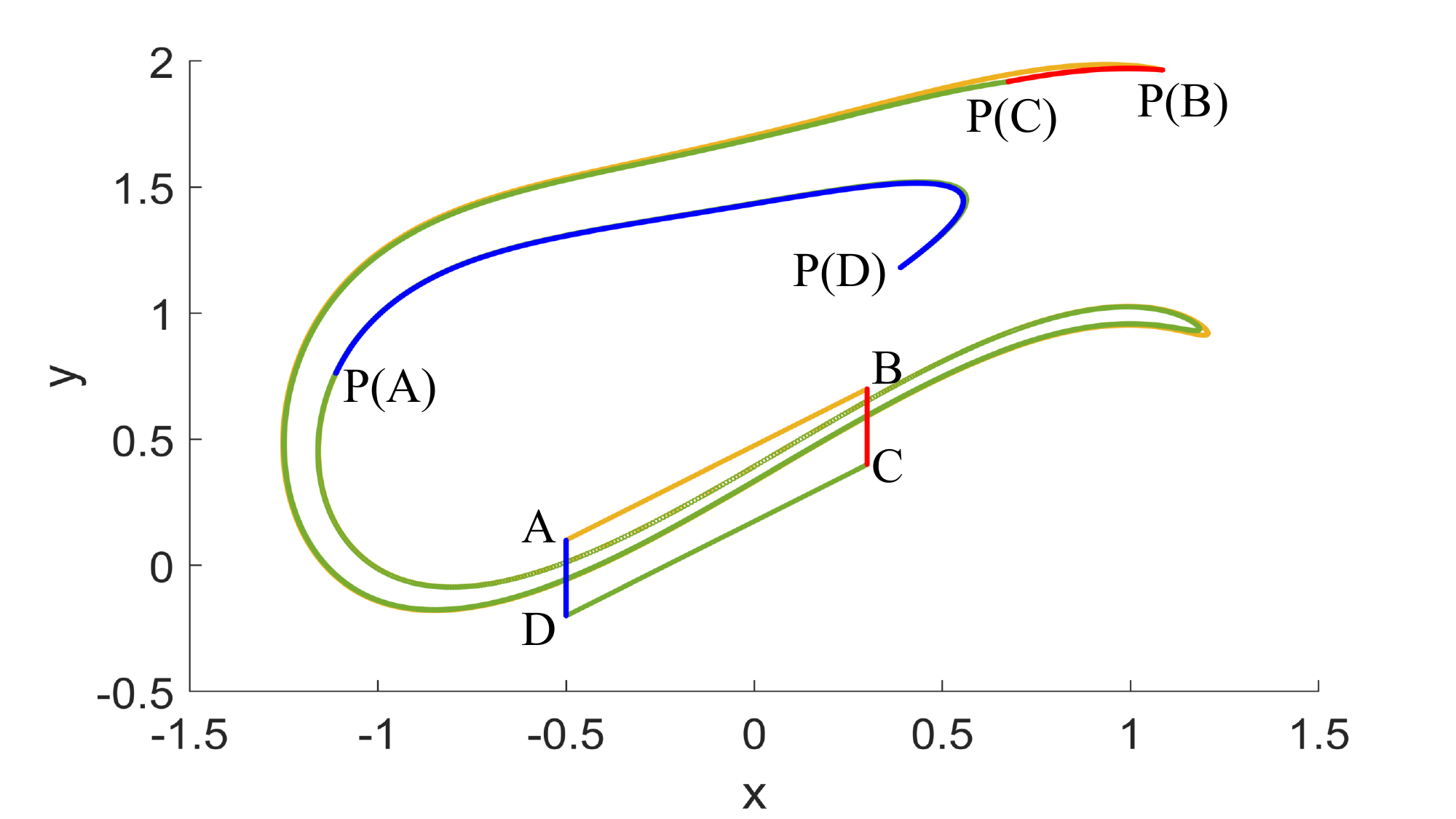}
	\caption{Smale horseshoe of the first return map when $\gamma=0.8$.}
	\label{fig2:d}
\end{figure}
\begin{figure}[H]
	\centering
	\includegraphics[width=\textwidth]{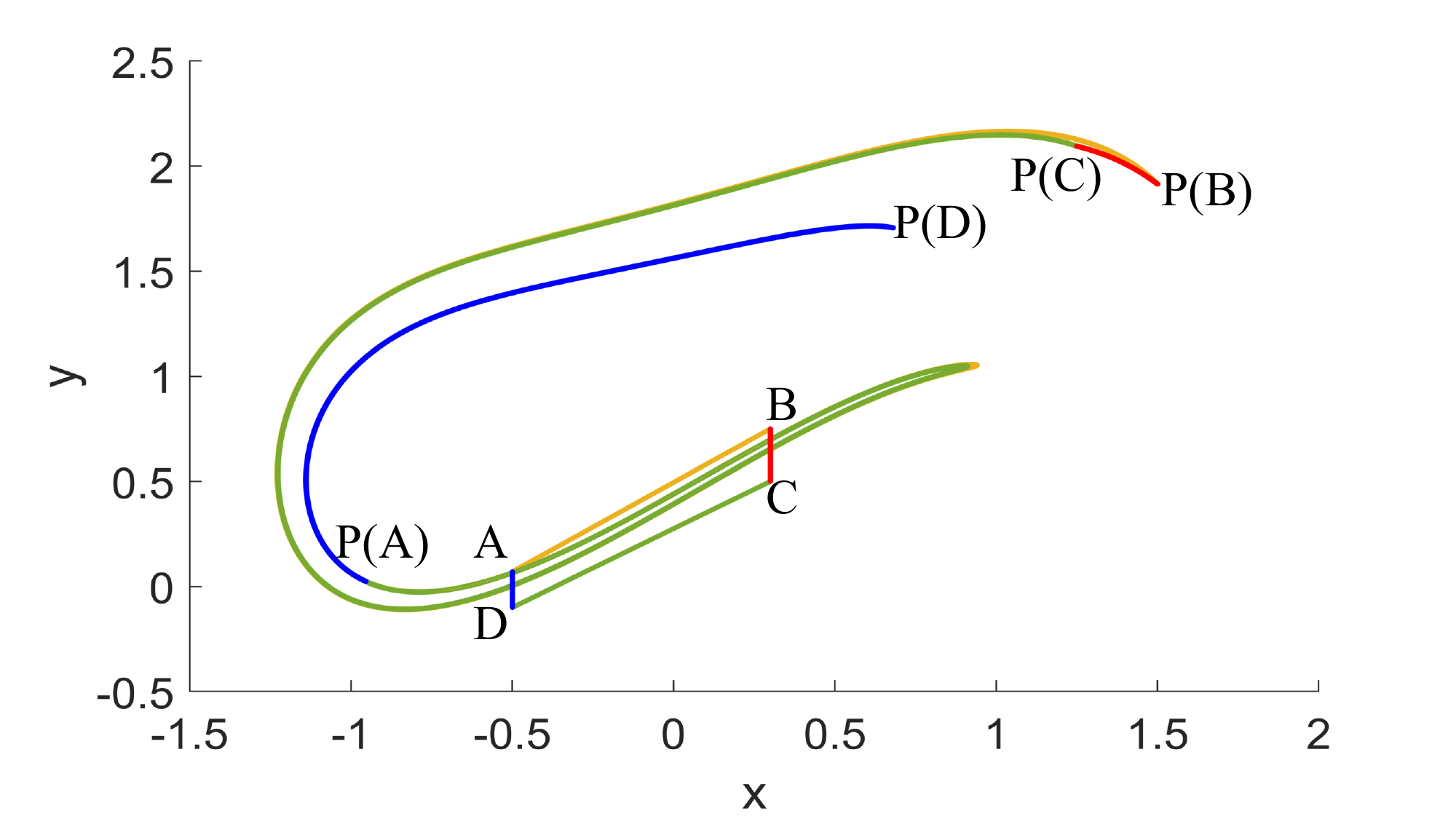}
	\caption{Smale horseshoe of the first return map when $\gamma=0.9$.}
	\label{fig2:e}
\end{figure}
\begin{figure}[H]
	\centering
	\includegraphics[width=\textwidth]{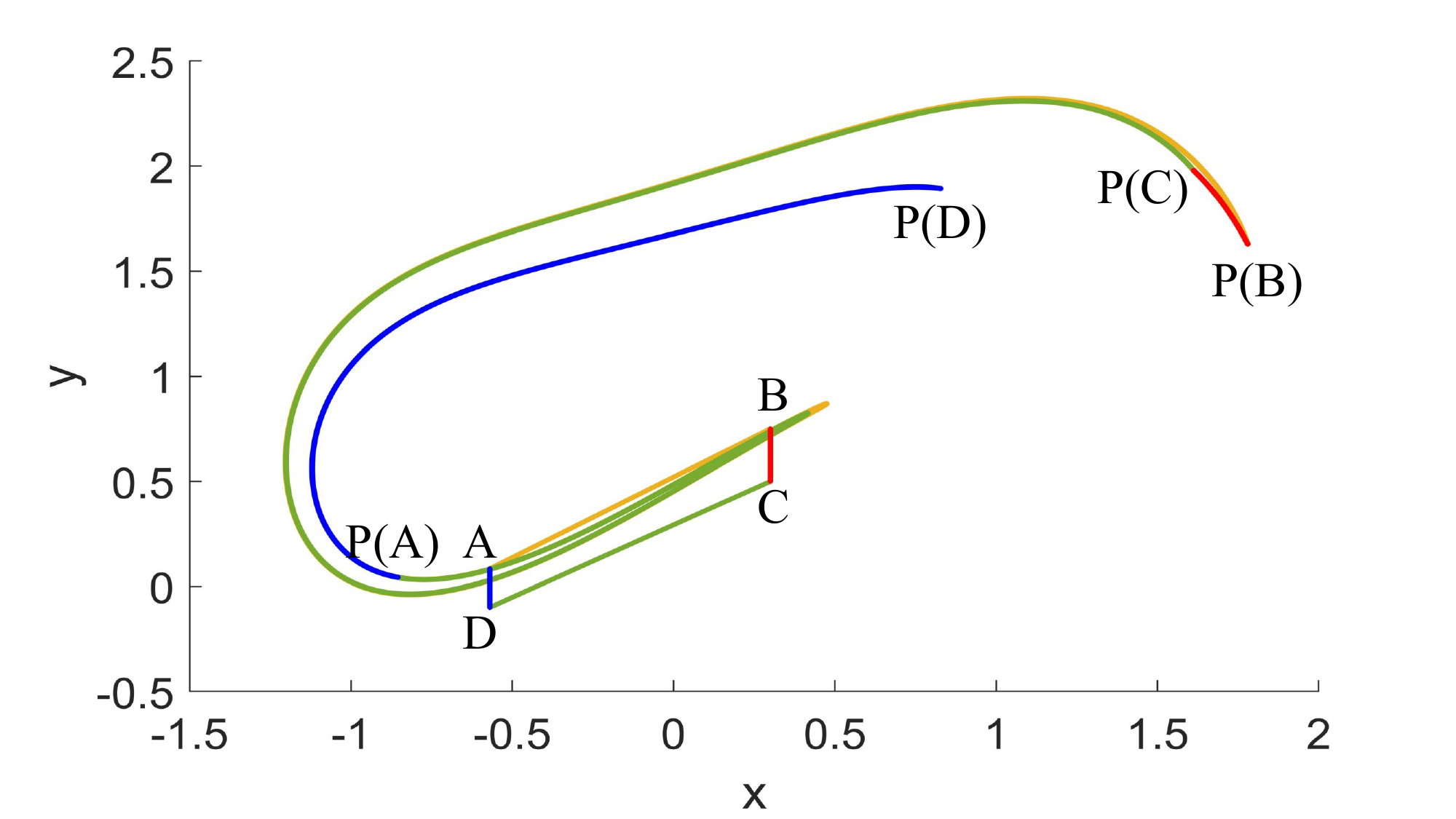}
	\caption{Smale horseshoe of the first return map when $\gamma=1.0$.}
	\label{fig2:f}
\end{figure}
\begin{figure}[H]
	\centering
	\includegraphics[width=\textwidth]{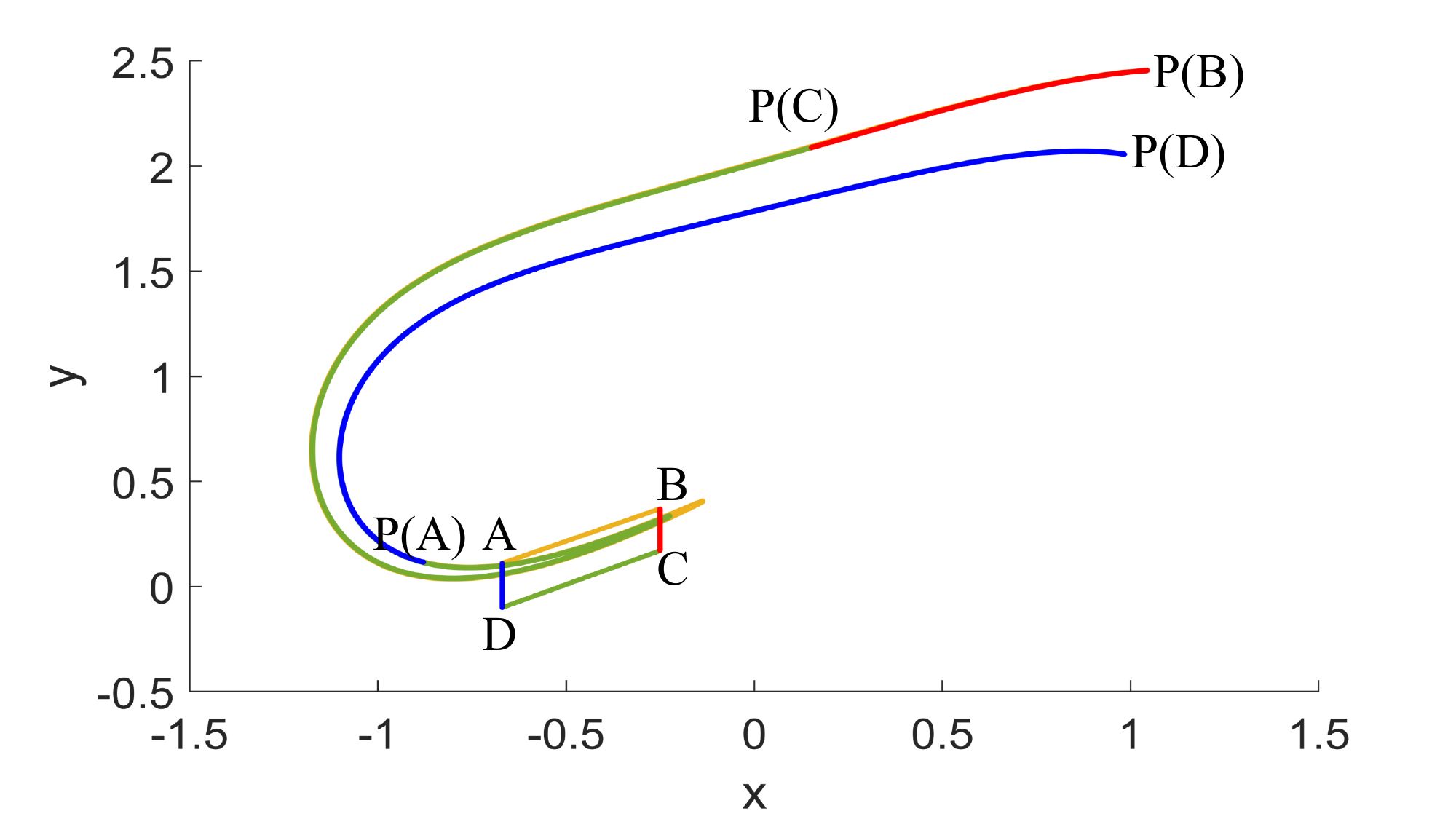}
	\caption{Smale horseshoe of the first return map when $\gamma=1.1$.}
	\label{fig2:g}
\end{figure}
\begin{figure}[H]
	\centering
	\includegraphics[width=\textwidth]{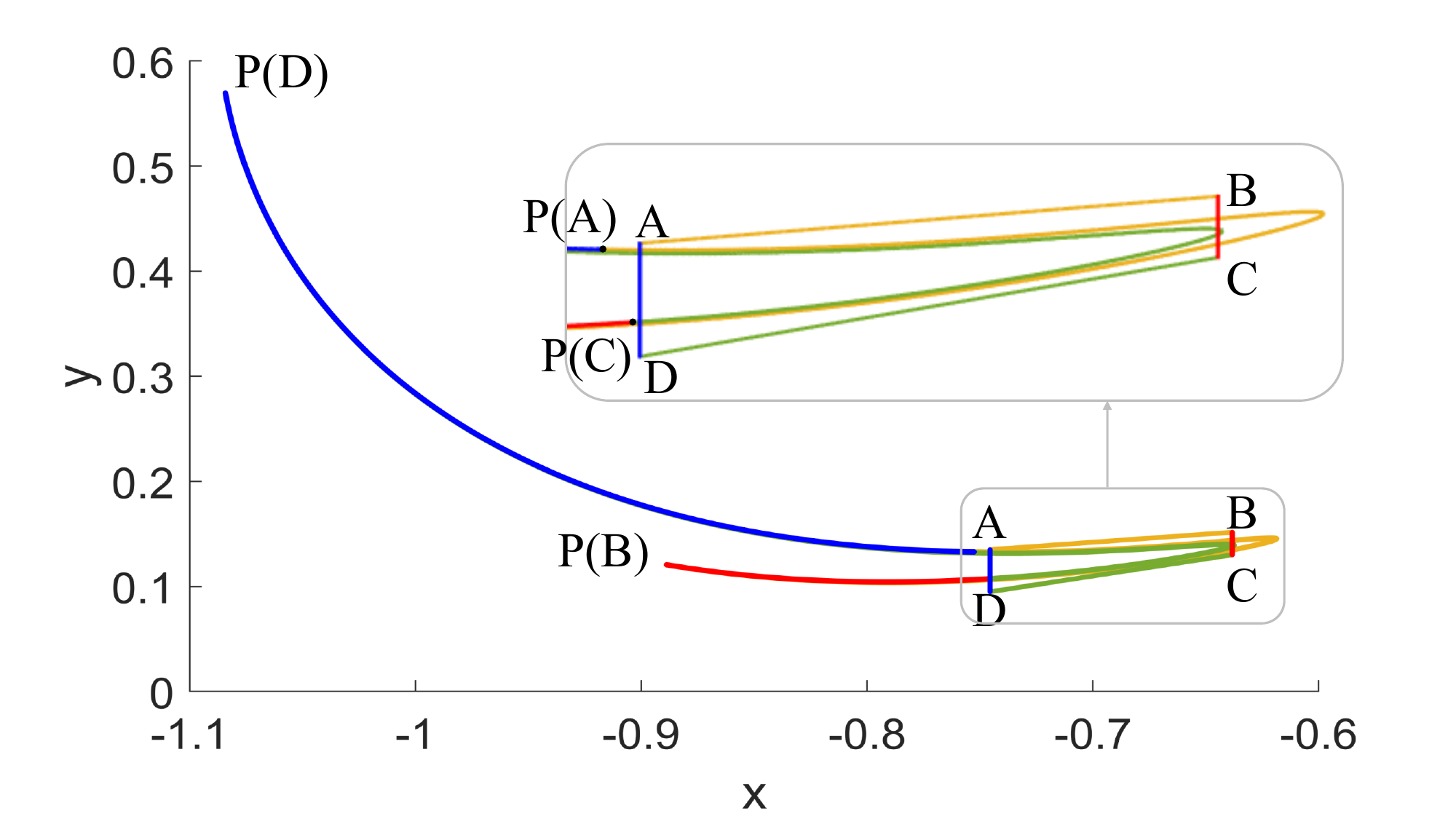}
	\caption{Smale horseshoe of the first return map when $\gamma=1.183$.}
	\label{fig2:h}
\end{figure}
It is found that for \( \gamma \in [0.4, 1.183] \), the Smale horseshoe persists. The coordinates of the four endpoints of the quadrilateral and the corresponding force amplitude are summarized in Table \ref{tab:table1}.
\begin{table}[h]
	\centering
	\caption{Summary of experimental results on the first return map.}
	\label{tab:table1}
	\begin{tabular}{lccccc}
		\toprule
		ID & $\gamma$   & $A$  & $B$  & $C$  & $D$  \\
		\midrule
		1 & $0.40$  & $(-0.23, 0.05)$ & $(0.13, 0.40)$ & $(0.13, 0.20)$ & $(-0.23, -0.15)$\\
		2 & $0.50$ & $(-0.50, -0.106~6)$ & $(0.30, 0.60)$ & $(0.30, 0.20)$ & $(-0.50, -0.406~6)$\\
		3 & $0.60$ & $(-0.50, 0.00)$ & $(0.30, 0.70)$ & $(0.30, 0.20)$ & $(-0.50, -0.406~6)$\\
		4 & $0.70$ & $(-0.50, 0.00)$ & $(0.30, 0.70)$ & $(0.30, 0.20)$ & $(-0.50, -0.406~6)$\\
		5 & $0.80$ & $(-0.50, 0.10)$ & $(0.30, 0.70)$ & $(0.30, 0.40)$ & $(-0.50, -0.20)$\\
		6 & $0.90$ & $(-0.50, 0.07)$ & $(0.30, 0.75)$ & $(0.30, 0.50)$ & $(-0.50, -0.10)$\\
		7 & $1.00$ & $(-0.57, 0.084)$ & $(0.30, 0.75)$ & $(0.30, 0.50)$ & $(-0.57, -0.10)$\\				
		8 & $1.10$ & $(-0.67, 0.11)$ & $(-0.25, 0.37)$ & $(-0.25, 0.17)$ & $(-0.67, -0.10)$\\				
		9 & $1.180$ & $(-0.77, 0.134)$ & $(-0.626~6, 0.151~6)$ & $(-0.626~6, 0.13)$ & $(-0.77, 0.08)$\\				
		10 & 1.183 & $(-0.745~6, 0.135)$ & $(-0.638~3, 0.151~6)$ & $(-0.638~3, 0.13)$ & $(-0.745~6, 0.095)$\\
		\bottomrule
	\end{tabular}
\end{table}

When $\gamma$ keeps increasing, we find that the the chaotic behavior weakens to some extent. To be precise, the lower bound of the topological entropy decreases when $\gamma$ increases from $1.183$ to $1.184$.  We use $P_{\gamma=a}$ to refer to the first return map when the force amplitude is set to be $a$ for the remainder of this paper.
When $\gamma=1.183$, Figure~\ref{fig2:h}  shows that the Smale horseshoe still exists, implying that the Poincar\'e map $P_{\gamma=1.183}$ is semi-conjugate to a 2-shift map. It follows that
\[h(P_{\gamma=1.183})\geq \log 2.\]
\begin{figure}[H]
	\centering
	\includegraphics[width=\textwidth]{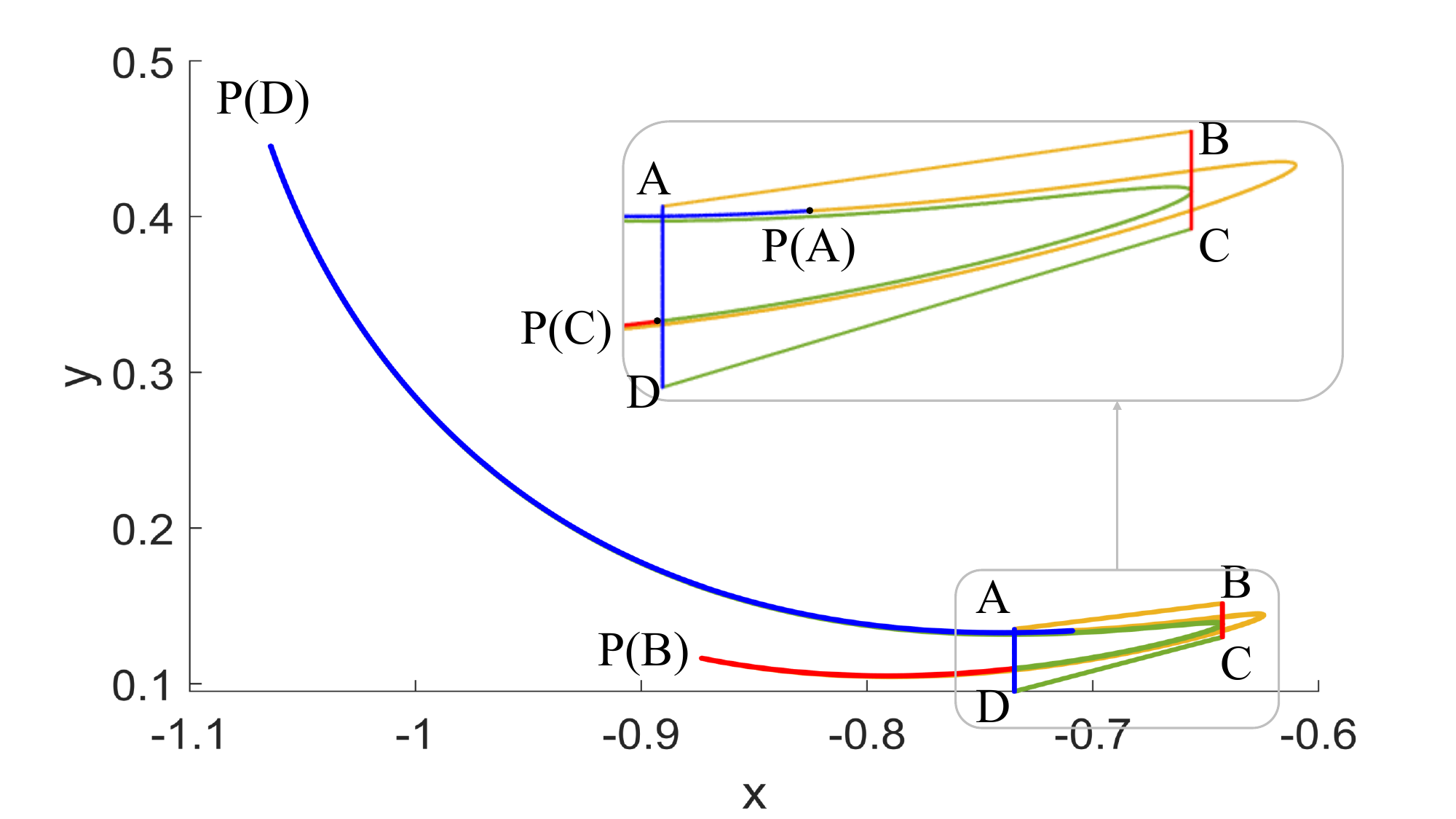}
	\caption{When $\gamma=1.184$, the Smale horseshoe of the first return map no longer exists.}
	\label{fig3}
\end{figure}
When $\gamma=1.184$, the Smale horseshoe no longer exists (Figure~\ref{fig3}), and it is difficult to find the topological horseshoe directly. Instead of searching for the ``complete'' horseshoe, we present an ``incomplete'' horseshoe in  Figure~\ref{fig4:a} with endpoints given by:
\begin{align*}
	Q_1: &A_1:(-0.734~8,0.135),~B_1:(-0.717~363~726,0.138~142~694),\\
	&C_1:(-0.669~605~101,0.119~775~478),~ D_1:(-0.734~8,0.095);\\
	Q_2: &A_2:(-0.688~220~372,0.143~395~46),~B_2:(-0.642~7,0.151~6),\\ &C_2:(-0.642~7,0.13),~D_2:(-0.668~665~587,0.120~132~513).
\end{align*}
To enhance comprehension, we present the geometric configuration of the horseshoe in Figure~\ref{fig4:b}, which shows that this is not the strict crossing structure defined in Section \ref{THT}. 
Nevertheless, in view of the Smale--Birkhoff theorem, iterating the return map can help us overcome this difficulty.
\begin{figure}[H]
	\centering
	\includegraphics[width=\textwidth]{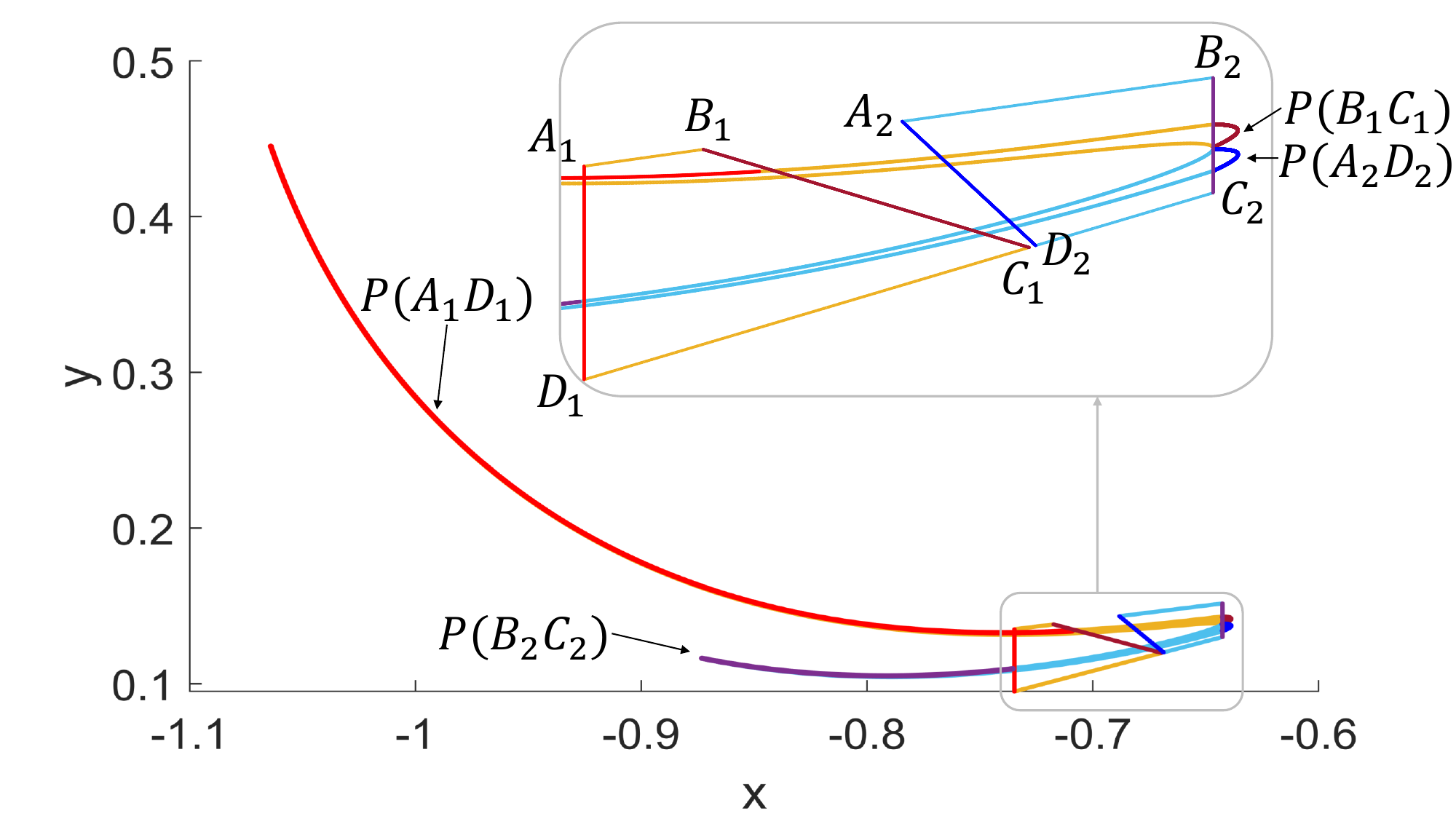}
	\caption{Pseudo-horseshoe of the first return map when $\gamma$=1.184.}
	\label{fig4:a}
\end{figure}
\begin{figure}[H]
	\centering
	\includegraphics[width=\textwidth]{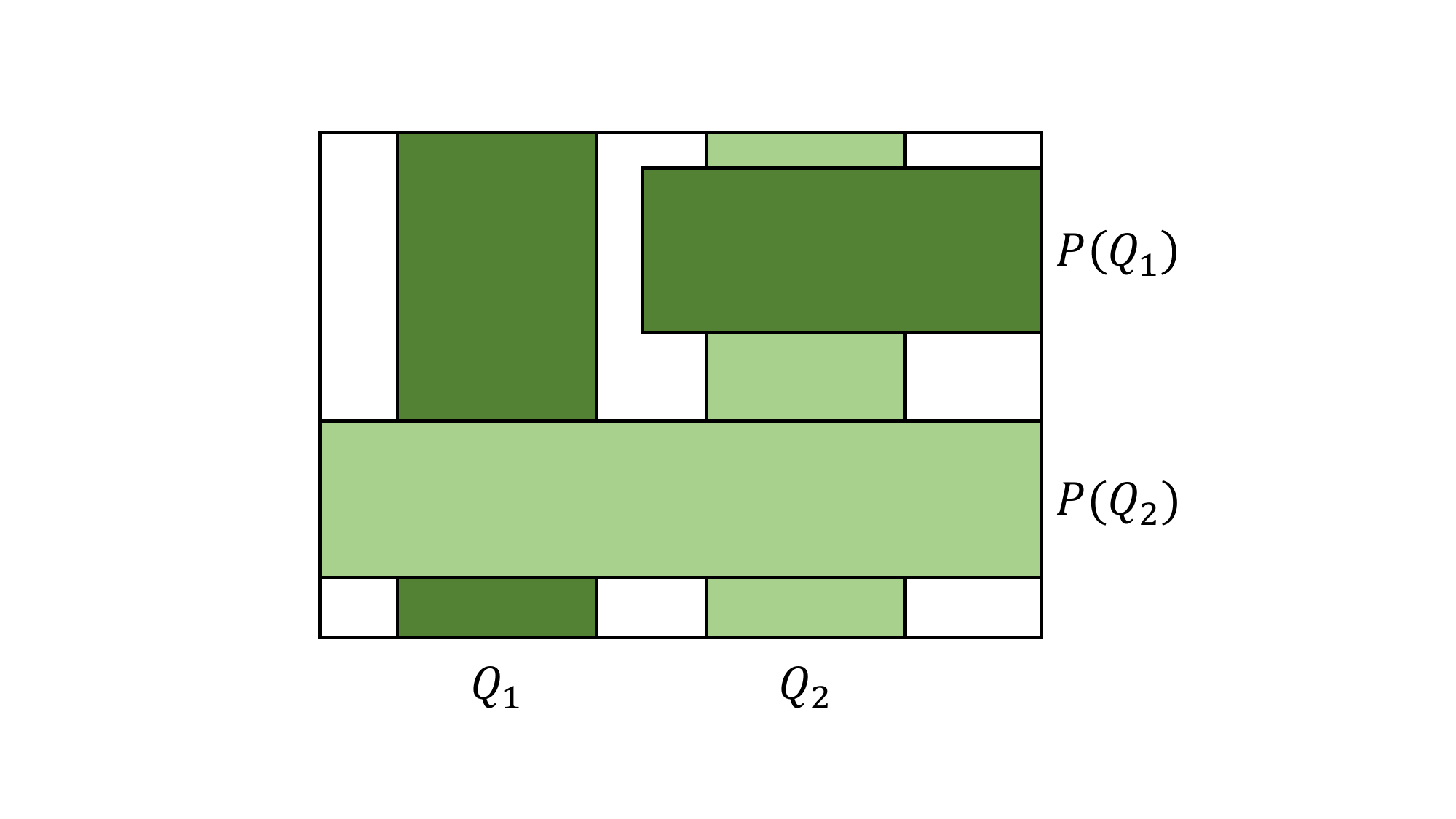}
	\caption{Schematic diagram of the pseudo-horseshoe in Figure~\ref{fig4:a}.}
	\label{fig4:b}
\end{figure}
To distinguish it from the topological horseshoe defined in Definition~\ref{thd}, we refer to the structure in Figure~\ref{fig4:b} as a pseudo-horseshoe. 
\begin{definition}
	Suppose $f: D \to \mathbb{R}^n$ satisfies the following assumptions:\\
	(1) There exist   $m$ mutually path-connected disjoint compact subsets $B_1,B_2,\dots$ and $B_m$ of $D$, and the restriction of $f$ to each $B_i$, denoted $f|_{B_i}$, is continuous;\\
	(2) For all $i\in\{1,2,\dots,m\}$, there exists $j\in\{1,2,\dots,m\}$, such that $f(B_i)\mapsto B_j$;\\
	(3) There exists a pair $(i,j) \in \{1,2,\dots,m\} \times \{1,2,\dots,m\}$ such that 
	$f(B_i) \cap B_j = \emptyset$;\\then we say $f$ has a pseudo-horseshoe, and such crossing structure is referred to be incomplete crossing.
\end{definition}
It is noteworthy that the existence of a pseudo-horseshoe does not necessarily imply the existence of chaos. More details can be found in Ref.~\cite{Cheng2025a}.
Fortunately, for the crossing structure in Fig~\ref{fig4:b}, we have the following fact.
\begin{theorem}\label{SHT}
	Let $X$ be a compact metric space, and $f: X \to X$ be a homeomorphism. Suppose $B_1$ and $B_2$ are disjoint, connected, nonempty, compact subsets of $X$, and satisfy the same crossing relationship shown in Figure~\ref{fig4:b}:
	\[ f(B_1) \mapsto B_2, \]
	\[ f(B_2) \mapsto B_j, \quad (j = 1, 2). \]
	Then, we have 
	\begin{eqnarray}
		h(f) \geq \log \frac{1 + \sqrt{5}}{2}.
	\end{eqnarray} 
\end{theorem}

\begin{proof}
	Set $B_{ij}=f(B_i)\cap B_j$, if $f(B_i)\cap B_j\neq \emptyset$. Then, we have 
	\[f(B_{12})\mapsto B_i, ~(i=1,2),\]	
	\[f(B_{22})\mapsto B_i, ~(i=1,2).\]	
	Recalling that $B_{12}\subset f(B_1)$, we deduce that
	\[f^2(f^{-1}(B_{12}))\mapsto B_i, (i=1,2).\]
	Similarly, we have
	\[f^2(f^{-1}(B_{22}))\mapsto B_i, (i=1,2).\]
	Furthermore, it is obvious that
	\[f^2(f^{-1}(B_{12}))\mapsto f^{-1}(B_{12}), \]
	\[f^2(f^{-1}(B_{12}))\mapsto f^{-1}(B_{22}),\]
	\[f^2(f^{-1}(B_{22}))\mapsto f^{-1}(B_{12}),\]
	\[f^2(f^{-1}(B_{22}))\mapsto f^{-1}(B_{22}).\]
	
	\begin{figure}[H]
		\centering
		\includegraphics[height=8cm]{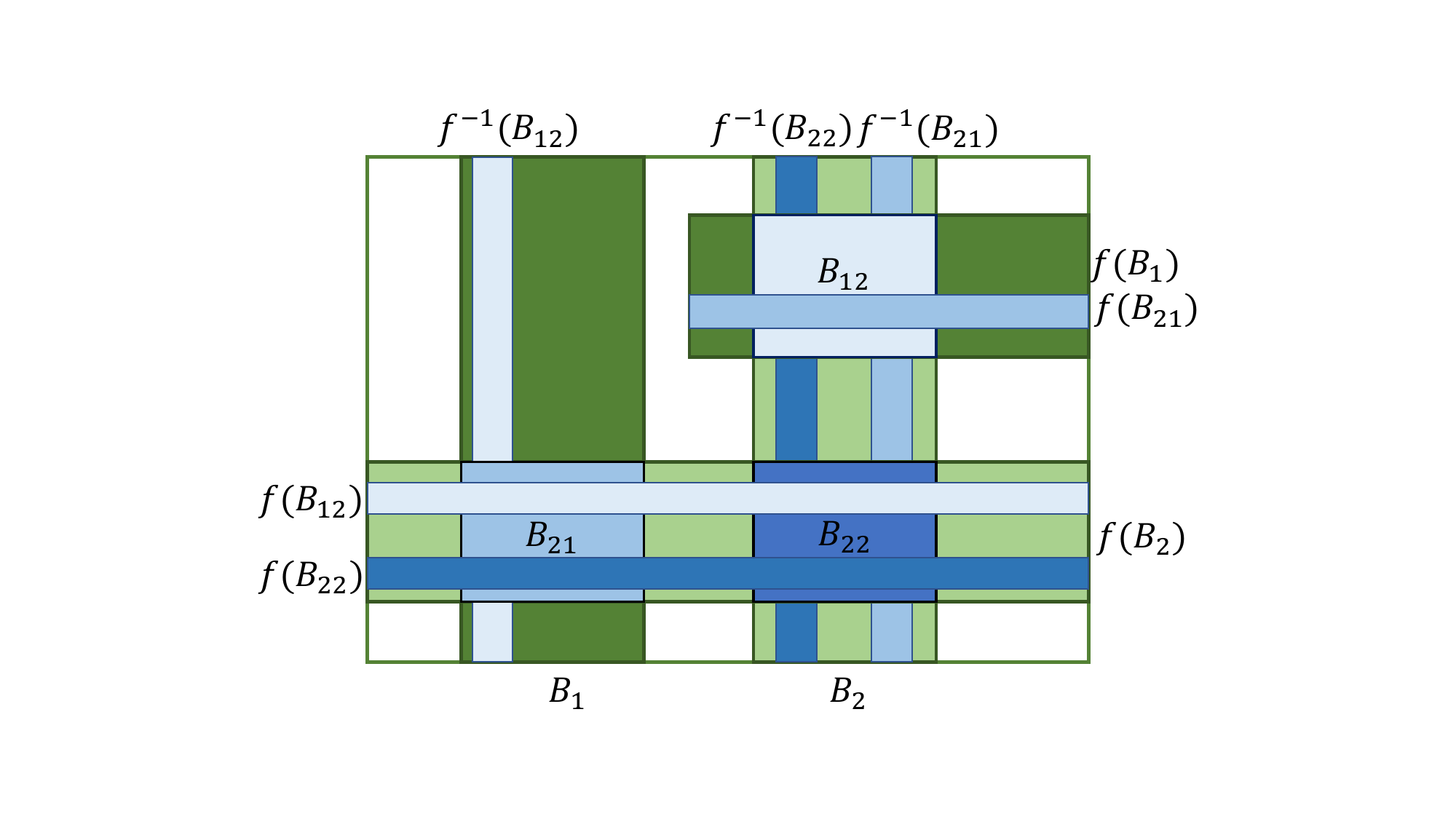}
		\caption{There are 2 crossing blocks with respect to $f^2$.}
		\label{ppt-2}
	\end{figure}

	It follows that $f^{-1}(B_{12})$ and $f^{-1}(B_{22})$ are 2 crossing blocks with respect to $f^2$, as illustrated  in Figure~\ref{ppt-2}. From Theorem \ref{TE}, there exists an invariant set $\Lambda\in X$, such that $f^2|_\Lambda$ is semi-conjugate to the 2-shift $\sigma$ and  $h(f^2)\geq \log 2$, where $h(f^2)$ is the topological entropy of $f^2$. For \( f \), Theorem \ref{TE} implies that \( h(f) = \frac{1}{2} h(f^2) \geq \frac{1}{2} \log 2 \).
	
	By simple induction, if we find $x_n$ blocks, whose images under \( f^n \) cross \( B_2 \), and \( y_n \) of their images cross \( B_1 \) when considering \( f^n \),
	then we will find \( x_n + y_n \) blocks whose images under \( f^{n+1} \) cross \( B_2 \), and \( x_n \) of their images crossing \( B_1 \) when considering \( f^{n+1} \).  This implies the existence of \( x_n \) crossing blocks with respect to \( f^{n+1} \). In other words, we obtain the following sequence defined by the recursive formulas:
	
	\begin{eqnarray}
		\left\{\begin{array}{l}
			x_{n+1}=x_n+y_n,\\
			y_{n+1}=x_n,\\
			x_1=2,\\
			y_1=1.
		\end{array}\right.
	\end{eqnarray}	
	Focusing on $\{x_n\}$, we get a Fibonacci sequence:
	\begin{eqnarray}
		\left\{\begin{array}{l}
			x_{n+2}=x_n+x_{n+1},\\
			x_1=2,\\
			x_2=3.
		\end{array}\right.
	\end{eqnarray}	
	Using the general term formula of the Fibonacci sequence, we have 
	\begin{eqnarray}
		x_n=\frac{1}{\sqrt{5}} [(\frac{1+\sqrt{5}}{2})^{n+2}-(\frac{1-\sqrt{5}}{2})^{n+2}].
	\end{eqnarray}
	From the discussion above, we find that there exist $x_{n-1}$ crossing blocks with respect to $f^n$, thus
	\begin{eqnarray}
		h(f)\geq\frac{1}{n}\log\{\frac{1}{\sqrt{5}} [(\frac{1+\sqrt{5}}{2})^{n+1}-(\frac{1-\sqrt{5}}{2})^{n+1}]\}.
	\end{eqnarray}
	
	Letting $n\to\infty$, we obtain
	\begin{eqnarray}
		h(f)\geq\log(\frac{1+\sqrt{5}}{2})\approx \log(1.61803)>\log(\sqrt{2})=\frac12\log2.	
	\end{eqnarray}
	
\end{proof}

Following the theorem above, we have the corollary below.

\begin{Cor}
	\begin{eqnarray}
		h(P_{\gamma=1.184})\geq\log(\frac{1+\sqrt{5}}{2}),
	\end{eqnarray}
	where $h(P_{\gamma=1.184})$ represents the topological entropy of the first return map $P_{\gamma=1.184}$ of the perturbed Duffing system when $\gamma=1.184$.
\end{Cor}

The procedure presented above can be conducted on the second return map $P^2$ and the third return map $P^3$ to get further estimation of the topological entropy. For $\gamma\geq 1.184$, several experiments have been conducted to find the Smale horseshoe of $P^2$ (Table \ref{tab:table2}).

\begin{table}[h]
	\caption{Summary of experimental results on the second return map.}
	\centering
	\label{tab:table2}
	\begin{tabular}{lccccc}
			\toprule
			ID & $\gamma$   & $A$  & $B$  & $C$  & $D$  \\
			\midrule
			1 & 1.184 & $(-0.676, 0.138)$ & $(-0.653, 0.140~27)$ & $(-0.653, 0.130)$ & $(-0.676, 0.122~4)$\\
			
			2 & 1.185 & $(-0.678, 0.138)$ & $(-0.653, 0.140~27)$ & $(-0.653, 0.130)$ & $(-0.678, 0.122~4)$\\
			
			3 & 1.187 & $(-0.685, 0.138)$ & $(-0.653~42, 0.140~27)$ & $(-0.653~42, 0.130)$ & $(-0.685, 0.122~4)$\\
			
			4 & 1.190 & $(-0.690, 0.137~188)$ & $(-0.666~78, 0.137~638)$ & $(-0.666~78, 0.133~146)$ & $(-0.690, 0.124~18)$\\
			
			5 & 1.191 & $(-0.690, 0.137~188)$ & $(-0.671~2, 0.137~638)$ & $(-0.671~2, 0.133~146)$ & $(-0.690, 0.124~18)$\\
			
			6 & 1.191~5 & $(-0.690, 0.136~667)$ & $(-0.673~5, 0.137~031)$ & $(-0.673~5, 0.133~479)$ & $(-0.690, 0.126~263)$\\
			
			7 & 1.191~7 & $(-0.690, 0.136~448)$ & $(-0.673~9, 0.137~031)$ & $(-0.673~9, 0.133~479)$ & $(-0.690, 0.126~263)$\\
			
			8 & 1.191~8 & $(-0.690, 0.136~356)$ & $(-0.674~29, 0.137~031)$ & $(-0.674~29, 0.133~479)$ & $(-0.690, 0.126~263)$\\
			
			9 & 1.191~9 & $(-0.690~442, 0.136~356)$ & $(-0.674~711, 0.136~754)$ & $(-0.674~711, 0.133~507)$ & $(-0.690~442, 0.126~819)$\\
			
			10 & 1.192~0 & $(-0.690~442, 0.136~356)$ & $(-0.675~138, 0.136~754)$ & $(-0.675~138, 0.133~69)$ & $(-0.690~442, 0.126~7)$\\
			\bottomrule
		\end{tabular}
\end{table}

\begin{figure}[H]
			\centering
	\includegraphics[width=\textwidth]{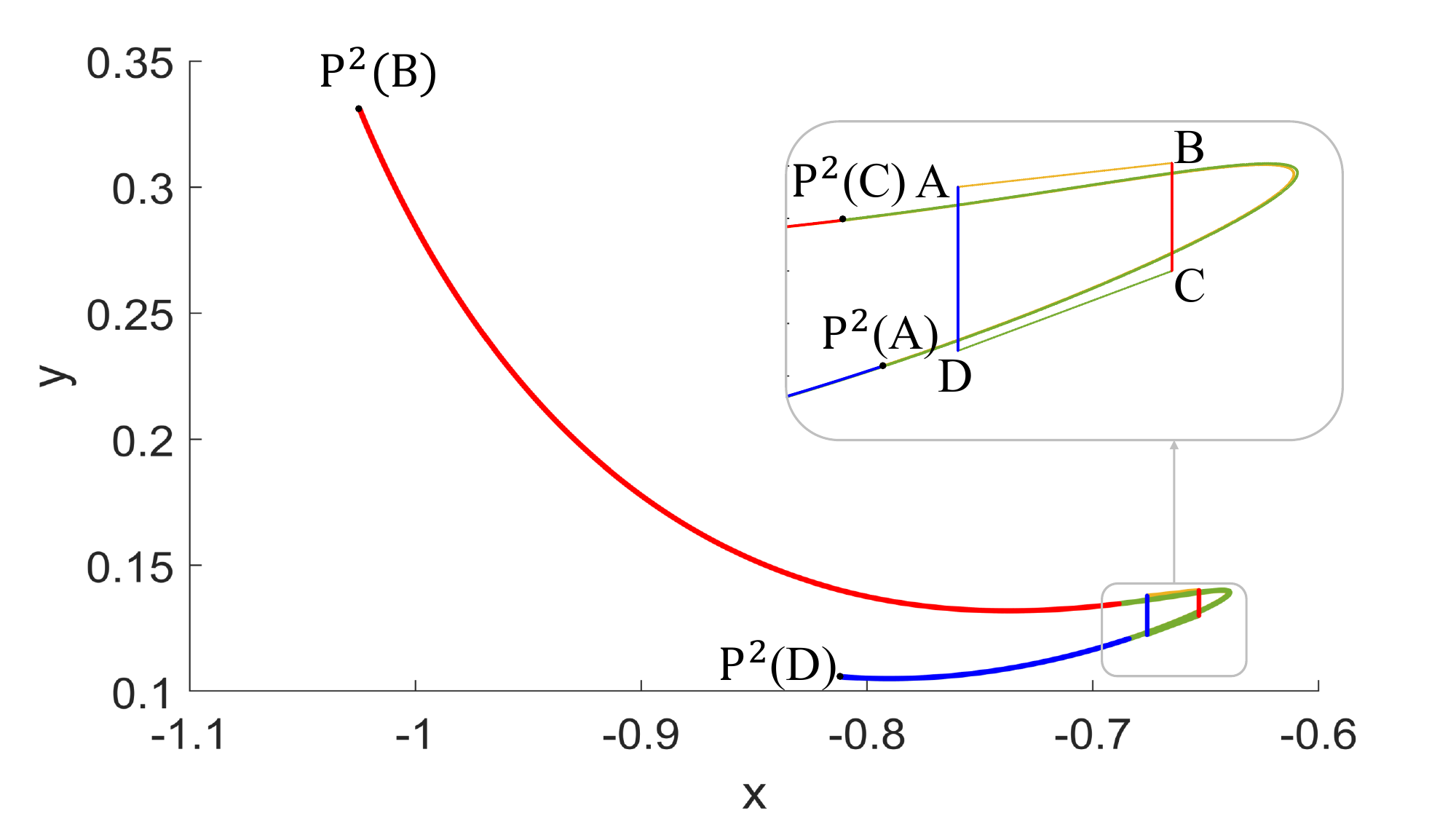}
	\caption{Smale horseshoe of the second return map when $\gamma=1.184$.}
	\label{fig6:a}
\end{figure}
\begin{figure}[H]
			\centering
	\includegraphics[width=\textwidth]{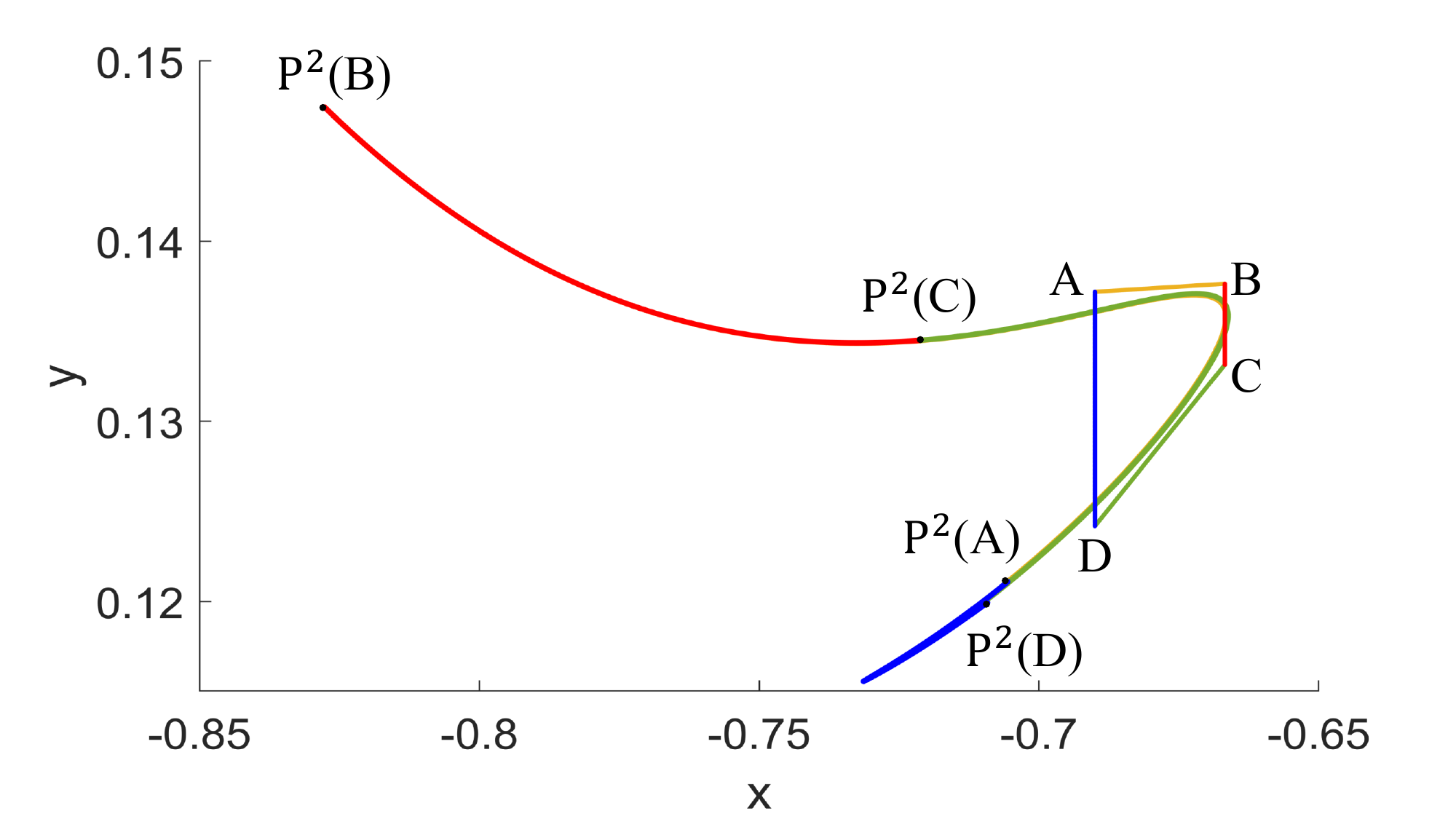}
	\caption{Smale horseshoe of the second return map when $\gamma=1.190$.}
	\label{fig6:b}
\end{figure}
\begin{figure}[H]
			\centering
	\includegraphics[width=\textwidth]{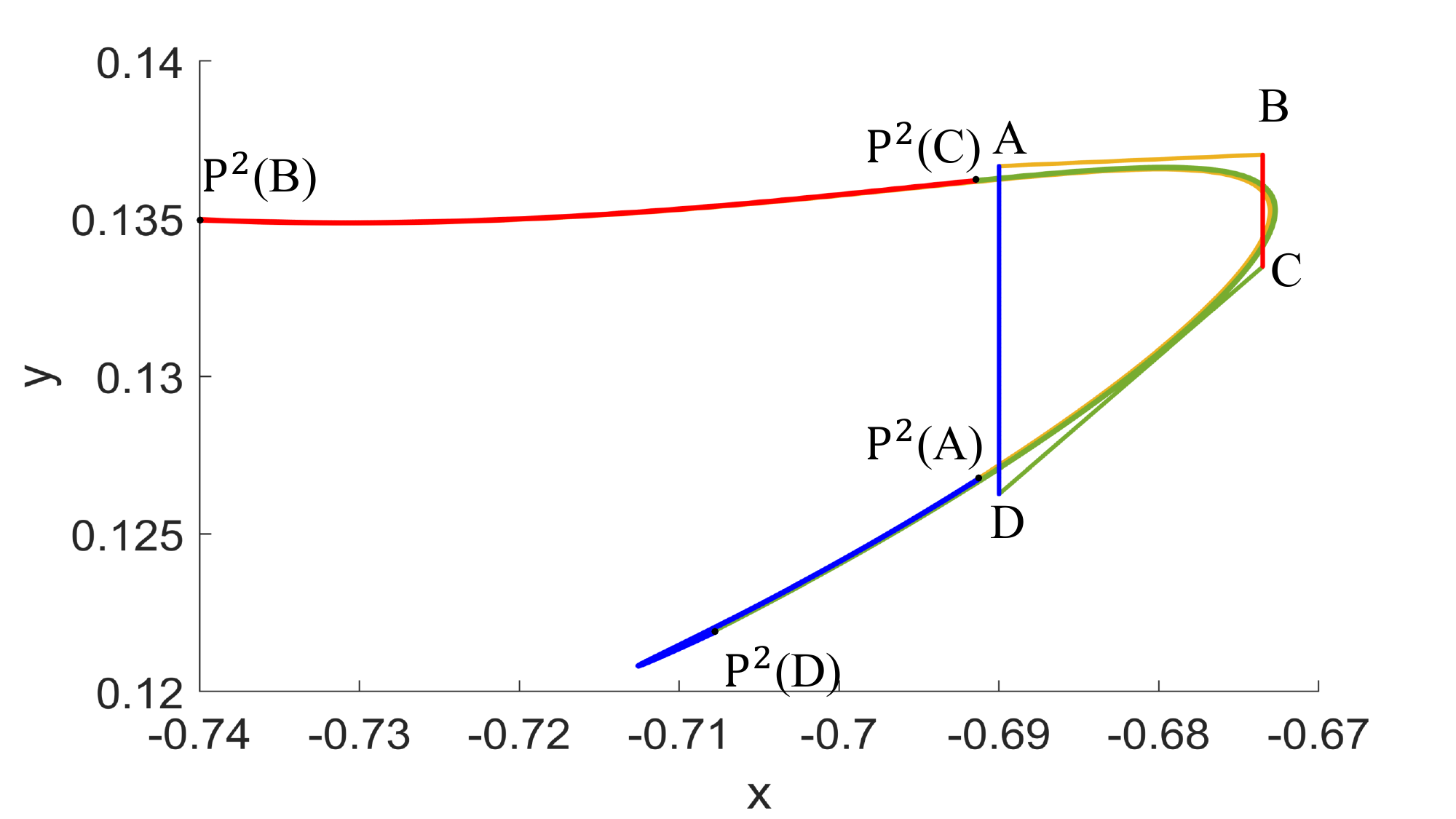}
	\caption{Smale horseshoe of the second return map when $\gamma=1.1915$.}
	\label{fig6:c}
\end{figure}
\begin{figure}[H]
			\centering
	\includegraphics[width=\textwidth]{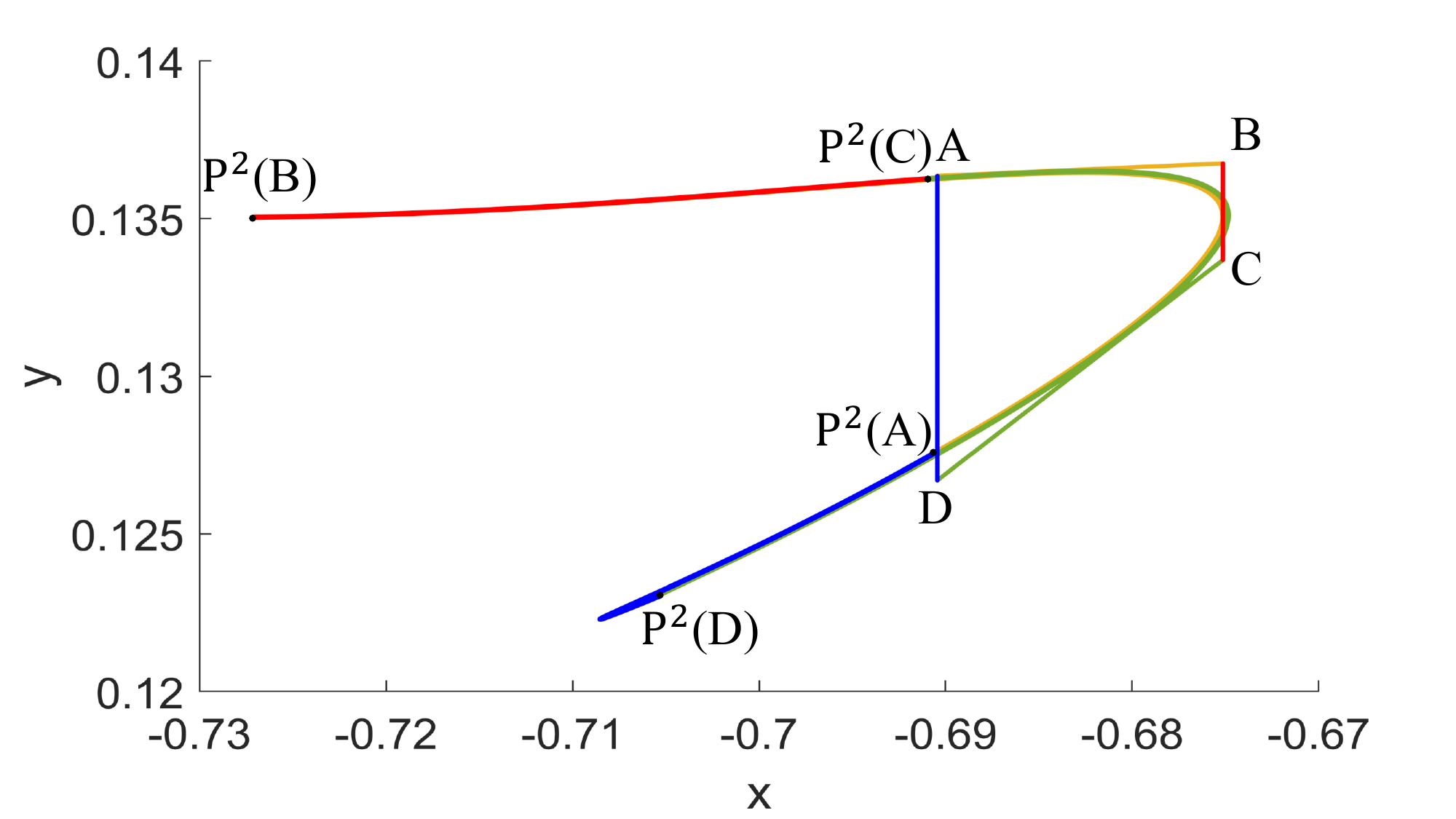}
	\caption{Smale horseshoe of the second return map when $\gamma=1.1920$.}
	\label{fig6:d}
\end{figure}

When $\gamma=1.1921$, the Smale horseshoe of $P^2_{\gamma=1.1921}$ disappears, and we find the incomplete crossing structure as in the case for $P_{\gamma=1.184}$ (Figure~\ref{fig7}).
\begin{figure}[H]
	\centering
	\includegraphics[width=\textwidth]{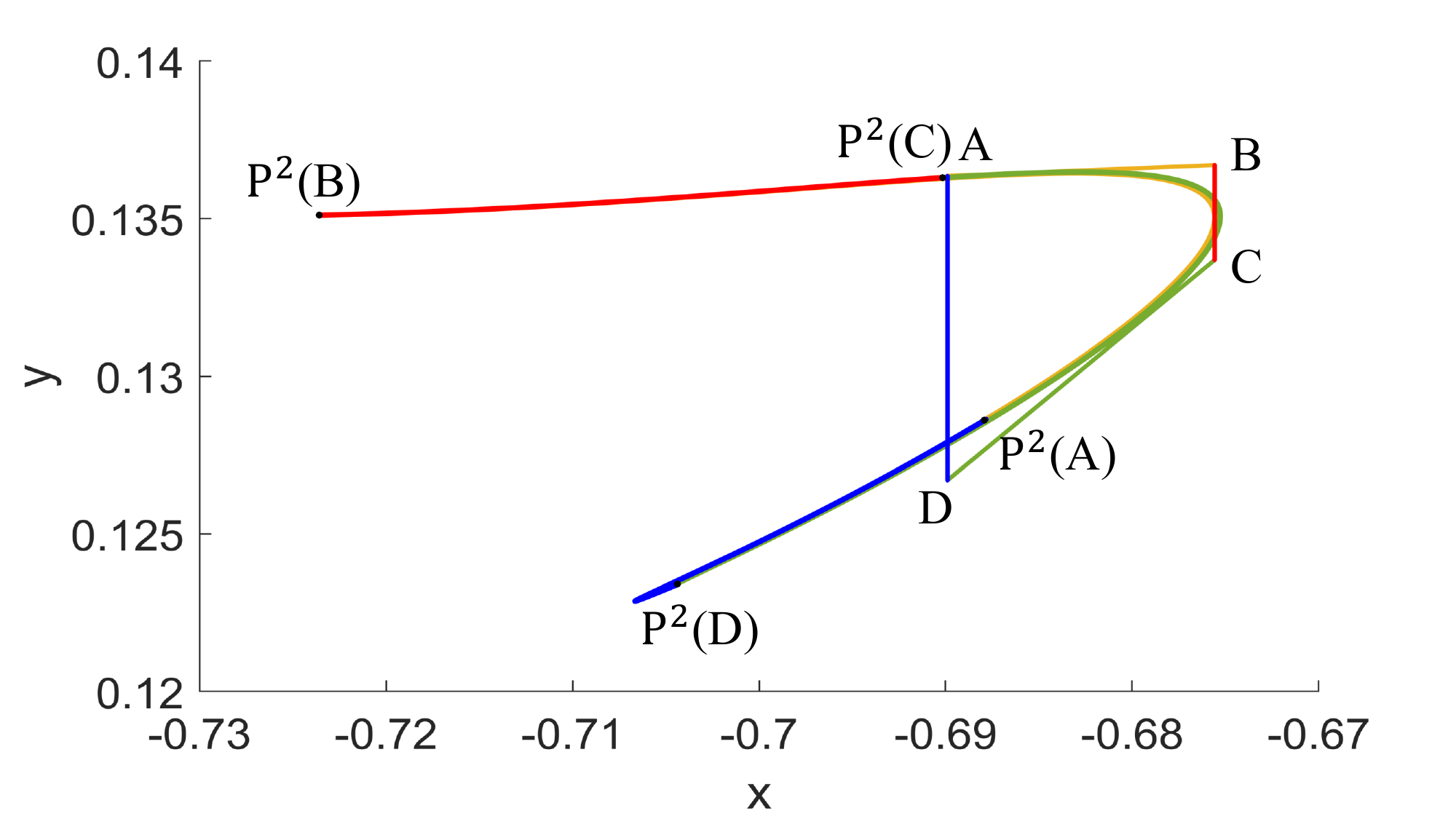}
	\caption{When $\gamma=1.1921$, the Smale horseshoe of the second return map disappears.}
	\label{fig7}
\end{figure}
Efforts have been made to find a pseudo-horseshoe (Figure~\ref{fig8:a}), whose endpoints are as follows.
\begin{align*}
	Q_1: &A_1:(-0.689~9,0.136~356),~B_1:(-0.684~3948~47,0.136~488~228),\\
	&C_1:(-0.679~477~483,0.131~786~817),~ D_1:(-0.689~9,0.126~7);\\
	Q_2: &A_2:(-0.683~993~429,0.136~498),~B_2:(-0.675~578,0.136~7),\\ 
	&C_2:(-0.675~578,0.133~69),~D_2:(-0.678~431,0.132~297~6).
\end{align*}

\begin{figure}[H]
			\centering
	\includegraphics[width=\textwidth]{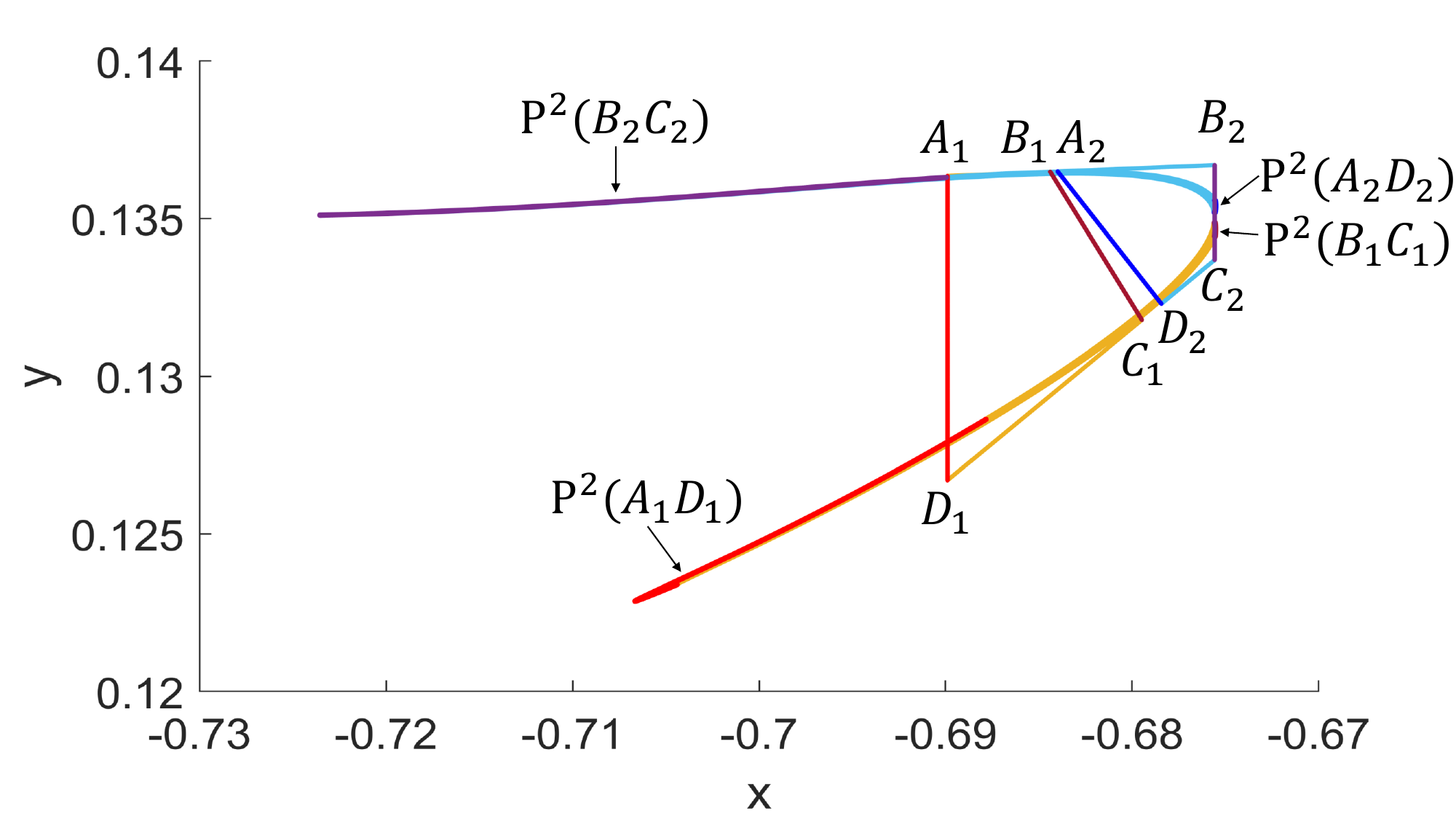}
	\caption{Pseudo-horseshoe of the second return map when $\gamma=1.1921$.}
	\label{fig8:a}
\end{figure}
\begin{figure}[H]
		\centering
\includegraphics[width=\textwidth]{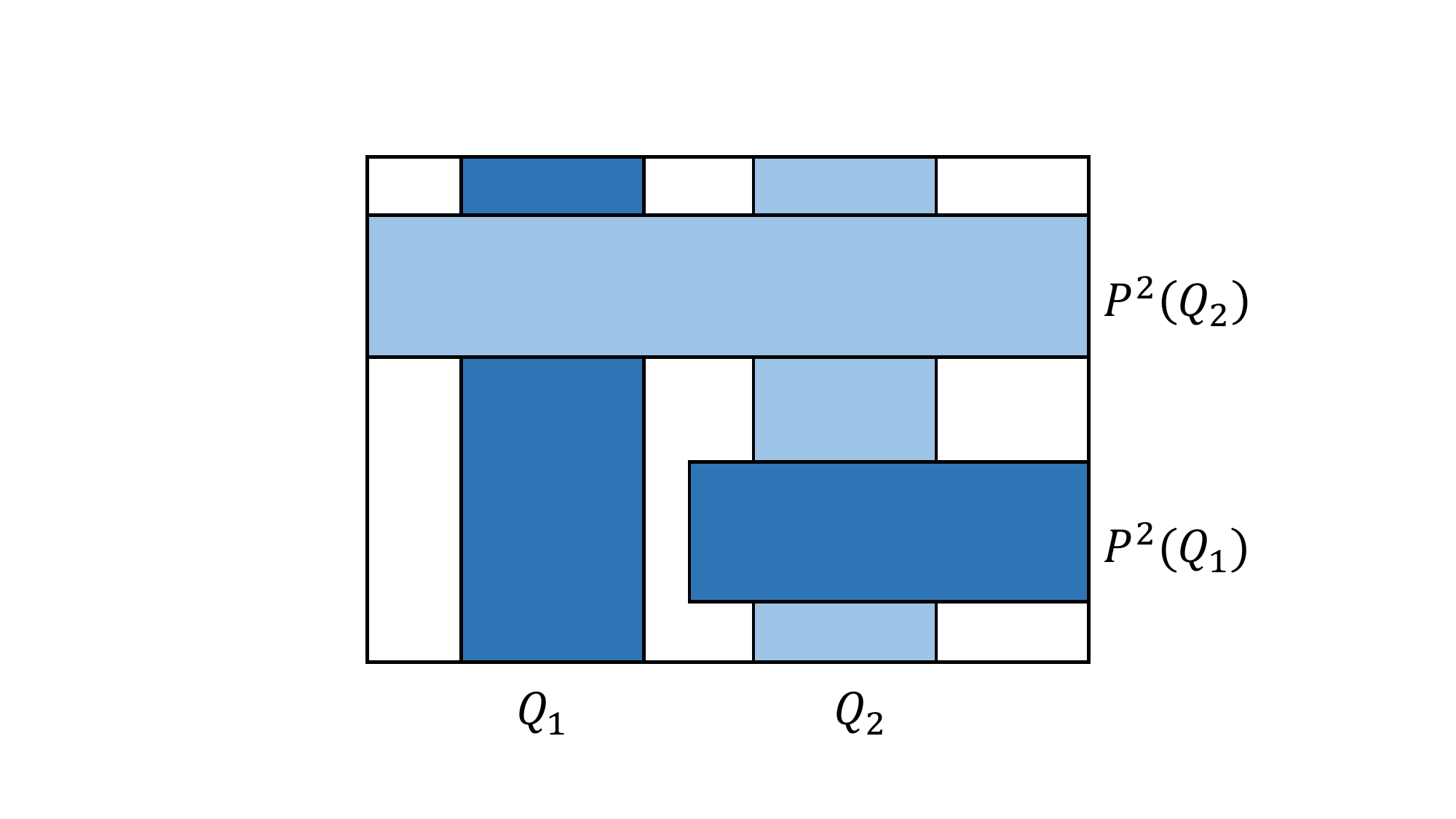}
\caption{Schematic diagram of the pseudo-horseshoe in Figure~\ref{fig8:a}.}
\label{fig8:b}
\end{figure}

The incomplete crossing structure of $P^2_{\gamma=1.1921}$   is shown in Figure~\ref{fig8:b}. Similar to the discussion on the structure in Figure~\ref{fig4:b}, we can conclude that 
\begin{eqnarray}
	h(P^2_{\gamma=1.1921})\geq \log \frac{1+\sqrt{5}}{2},
\end{eqnarray}
and thus	
\begin{eqnarray}
	h(P_{\gamma=1.1921})\geq\frac12\log\frac{1+\sqrt{5}}{2}.
\end{eqnarray}

To show that the procedure above can be conducted repeatably, we consider the third return map when $\gamma\geq1.1921$, and obtain the results shown in Table \ref{tab:table3}.

\begin{table}[h]
	\caption{Summary of experimental results on the third return map.}\label{tab:table3}
		\begin{tabular}{lccccc}
			\toprule
			ID & $\gamma$   & $A$  & $B$  & $C$  & $D$  \\
			\midrule
			1 & $1.192~1$ & $(-0.700, 0.136~4)$ & $(-0.680, 0.136~65)$ & $(-0.676, 0.132~4)$ & $(-0.689~8, 0.127~6)$\\
			
			2 & $1.192~3$ & $(-0.700, 0.136~4)$ & $(-0.680, 0.136~65)$ & $(-0.676, 0.132~4)$ & $(-0.690, 0.127~794)$\\
			
			3 & $1.192~5$ & $(-0.698, 0.136~4)$ & $(-0.680, 0.136~65)$ & $(-0.676~8, 0.133)$ & $(-0.690~224, 0.128~035)$\\
			
			4 & $1.192~7$ & $(-0.696, 0.136~4)$ & $(-0.680, 0.136~65)$ & $(-0.677~124, 0.133~4)$ & $(-0.691, 0.128)$\\
			
			5 & $1.192~8$ & $(-0.695~525, 0.136~2)$ & $(-0.680, 0.136~65)$ & $(-0.678~3, 0.132~6)$ & $(-0.693, 0.127~6)$\\
			\bottomrule
		\end{tabular}
\end{table}

Several experiments have been done to detect the Smale horseshoes of the third return map when $\gamma$ ranges in 1.1921 and 1.1928. These Smale horseshoes are shown in Figure~\ref{fig9:a}--Figure~\ref{fig9:c}.

\begin{figure}[H]
	\centering
	\includegraphics[width=\textwidth]{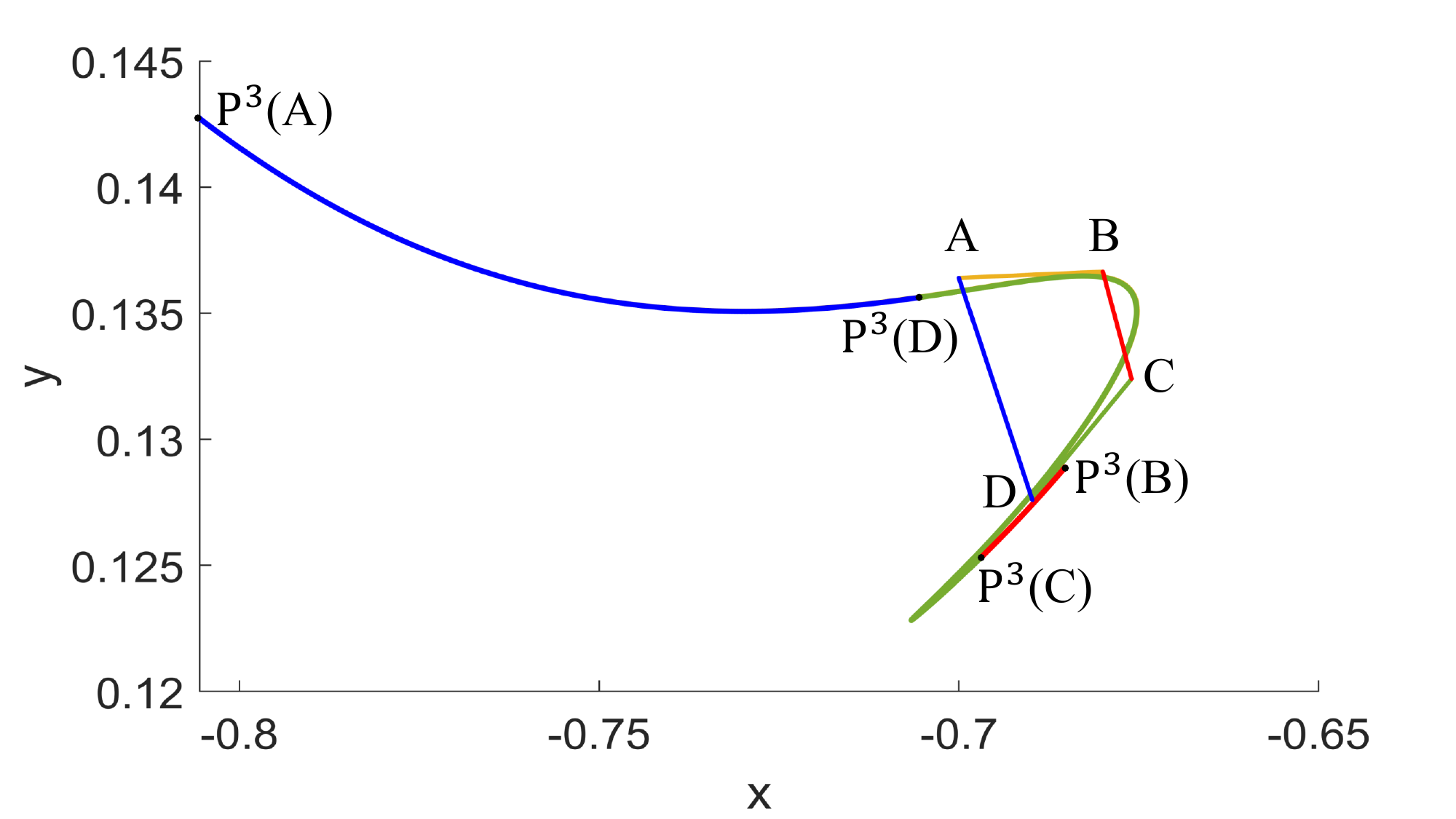}
	\caption{Smale horseshoe of the third return map when $\gamma=1.1921$.}
	\label{fig9:a}
\end{figure}
\begin{figure}[H]
	\centering
	\includegraphics[width=\textwidth]{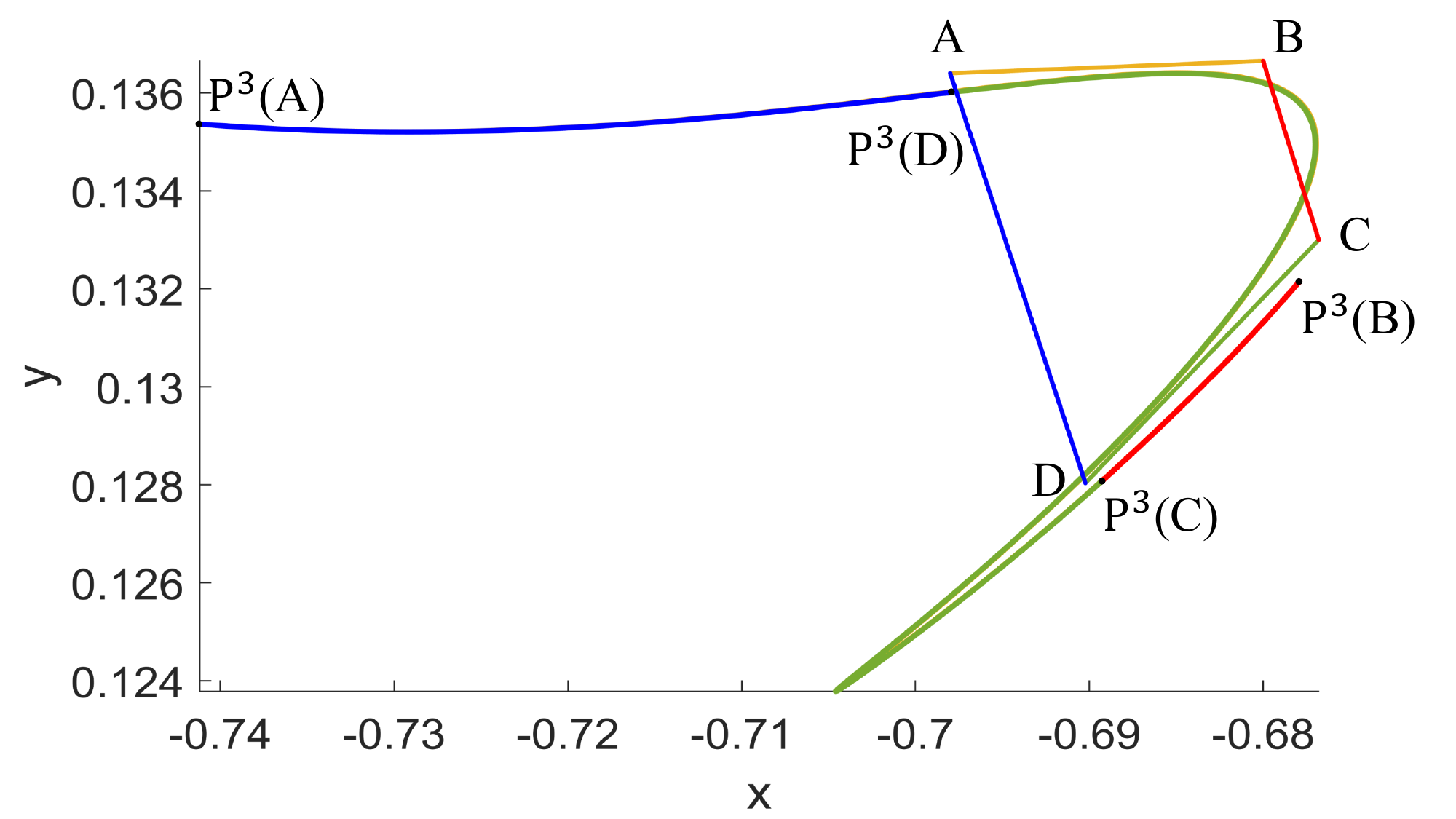}
	\caption{Smale horseshoe of the third return map when $\gamma=1.1925$.}
	\label{fig9:b}
\end{figure}
\begin{figure}[H]
	\centering
	\includegraphics[width=\textwidth]{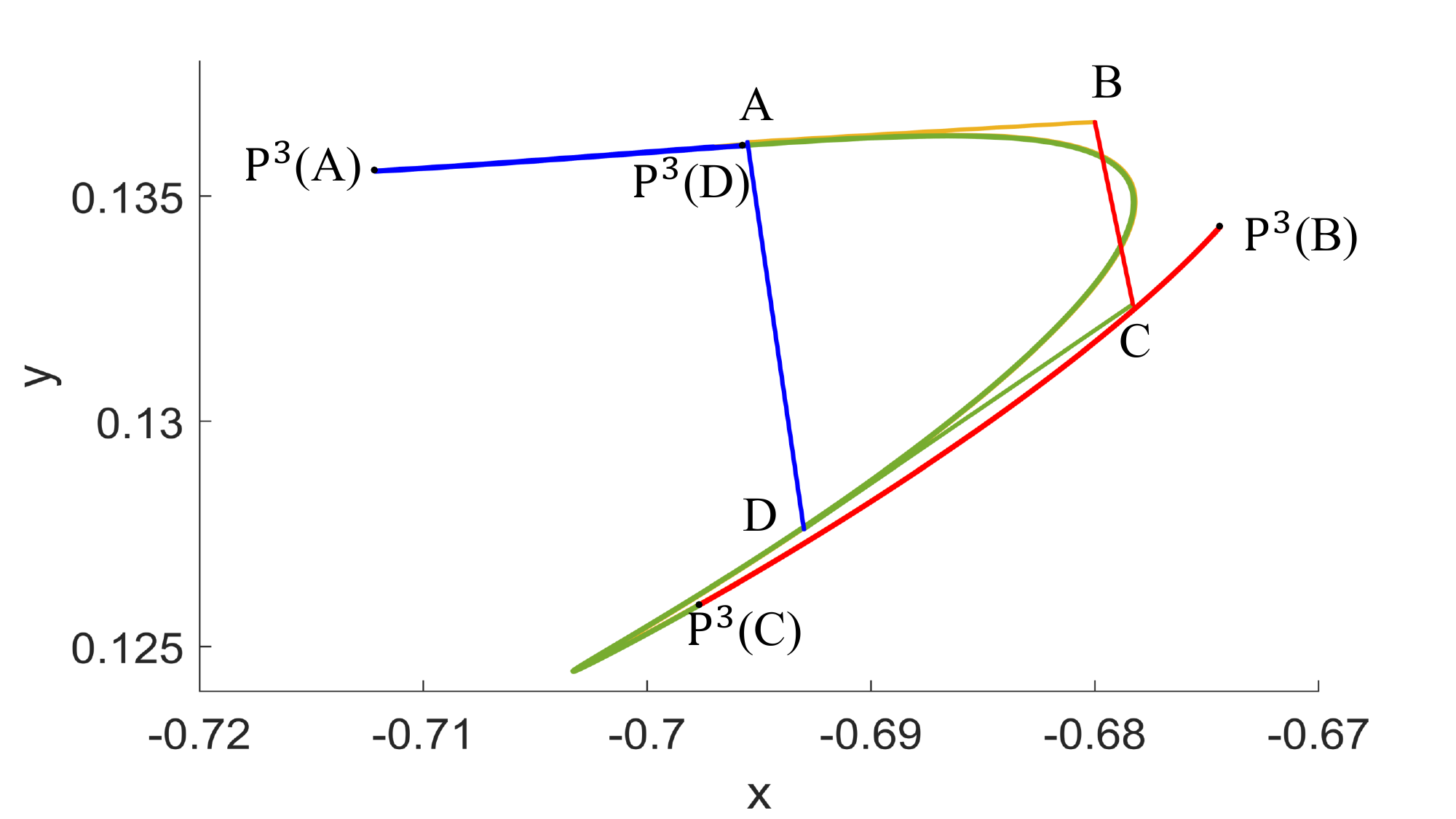}
	\caption{Smale horseshoe of the third return map when $\gamma=1.1928$.}
	\label{fig9:c}
\end{figure}
\begin{figure}[H]
	\centering
	\includegraphics[width=\textwidth]{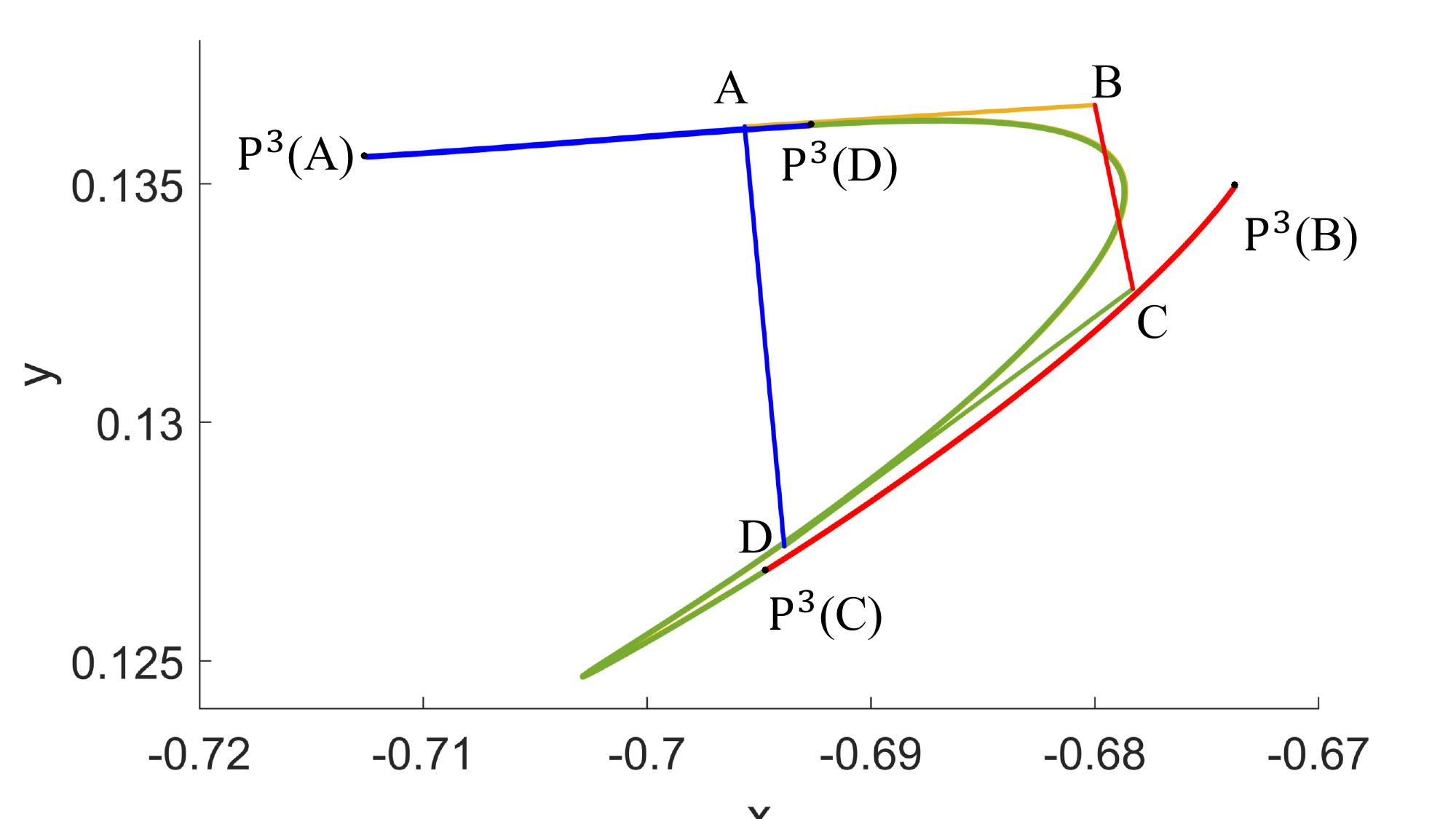}
	\caption{Smale horseshoe of the third return map disappears when $\gamma=1.1929$.}
	\label{fig9:d}
\end{figure}

When $\gamma=1.1929$, a Smale horseshoe is no longer observed (Figure~\ref{fig9:d}), while the incomplete crossing structure suggests the possible existence of a pseudo-horseshoe.

After several attempts, we find a pseudo-horseshoe of the third return map when $\gamma=1.1929$ (Figure~\ref{fig10:a}). The corresponding coordinates of the endpoints of the selected blocks are as follows.
\begin{align*}
	Q_1: &A_1:(-0.69565,0.1362),~B_1:(-0.69293984,0.136277928),\\
	&C_1:(-0.6869833,0.12979031),~ D_1:(-0.6938789,0.1274);\\
	Q_2: &A_2:(-0.691717918,0.136313063),~B_2:(-0.680317,0.136641),\\ &C_2:(-0.6783,0.1328),~D_2:(-0.668~665~587,0.130209).
\end{align*}

To enhance comprehension, we present the geometric configuration of the strange horseshoe in Figure~\ref{fig10:b}, which implies that 
\begin{eqnarray}
	h(P^3_{\gamma=1.1929})\geq \log\frac{1+\sqrt{5}}{2},
\end{eqnarray}
hence,
\begin{eqnarray}
	h(P_{\gamma=1.1929})\geq \frac13\log\frac{1+\sqrt{5}}{2}.	
\end{eqnarray}

\begin{figure}[H]
	\centering
	\includegraphics[width=\textwidth]{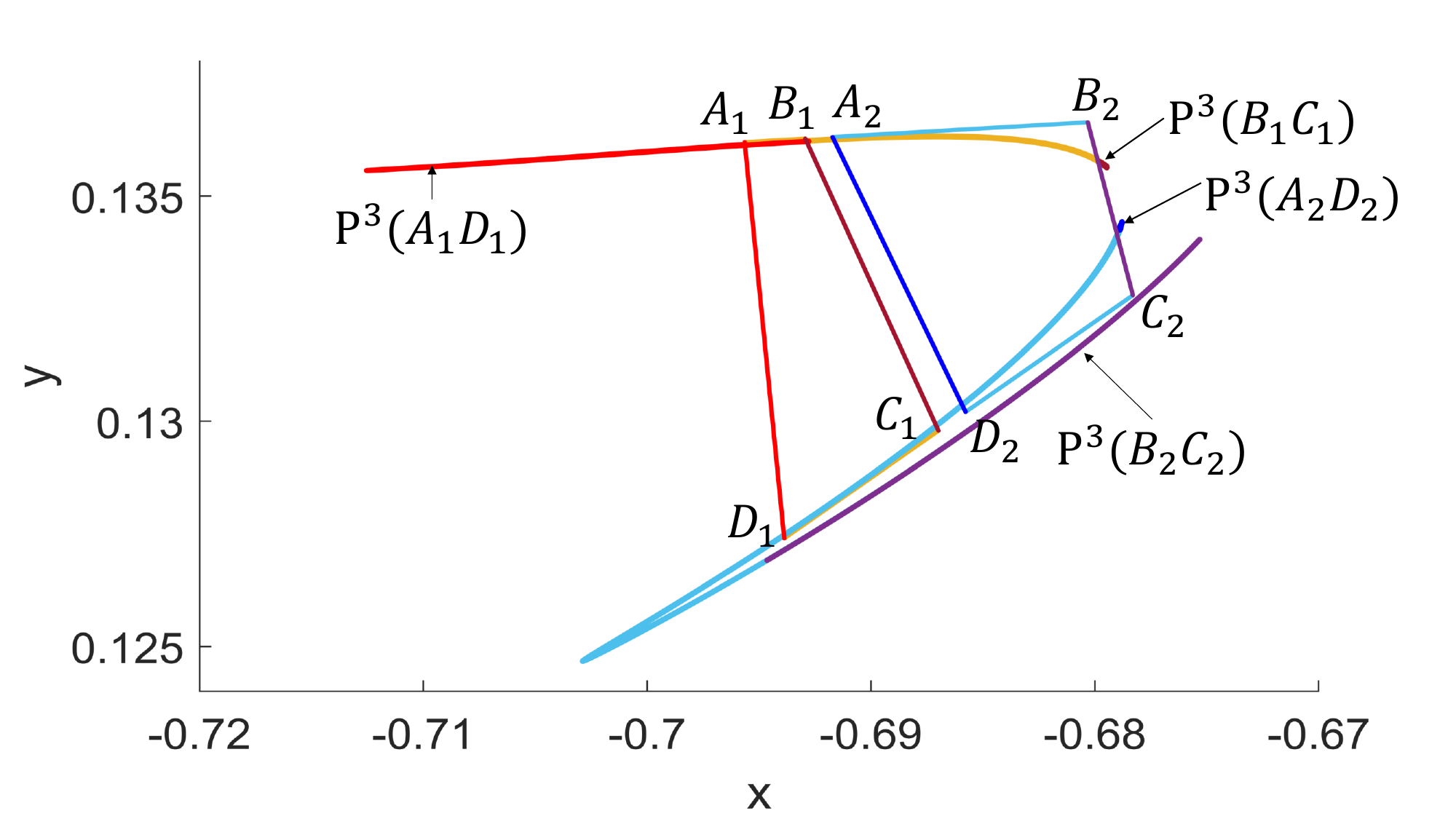}
	\caption{Pseudo-horseshoe of the third return map when $\gamma=1.1929$.}
	\label{fig10:a}
\end{figure}
\begin{figure}[H]
	\centering
	\includegraphics[width=\textwidth]{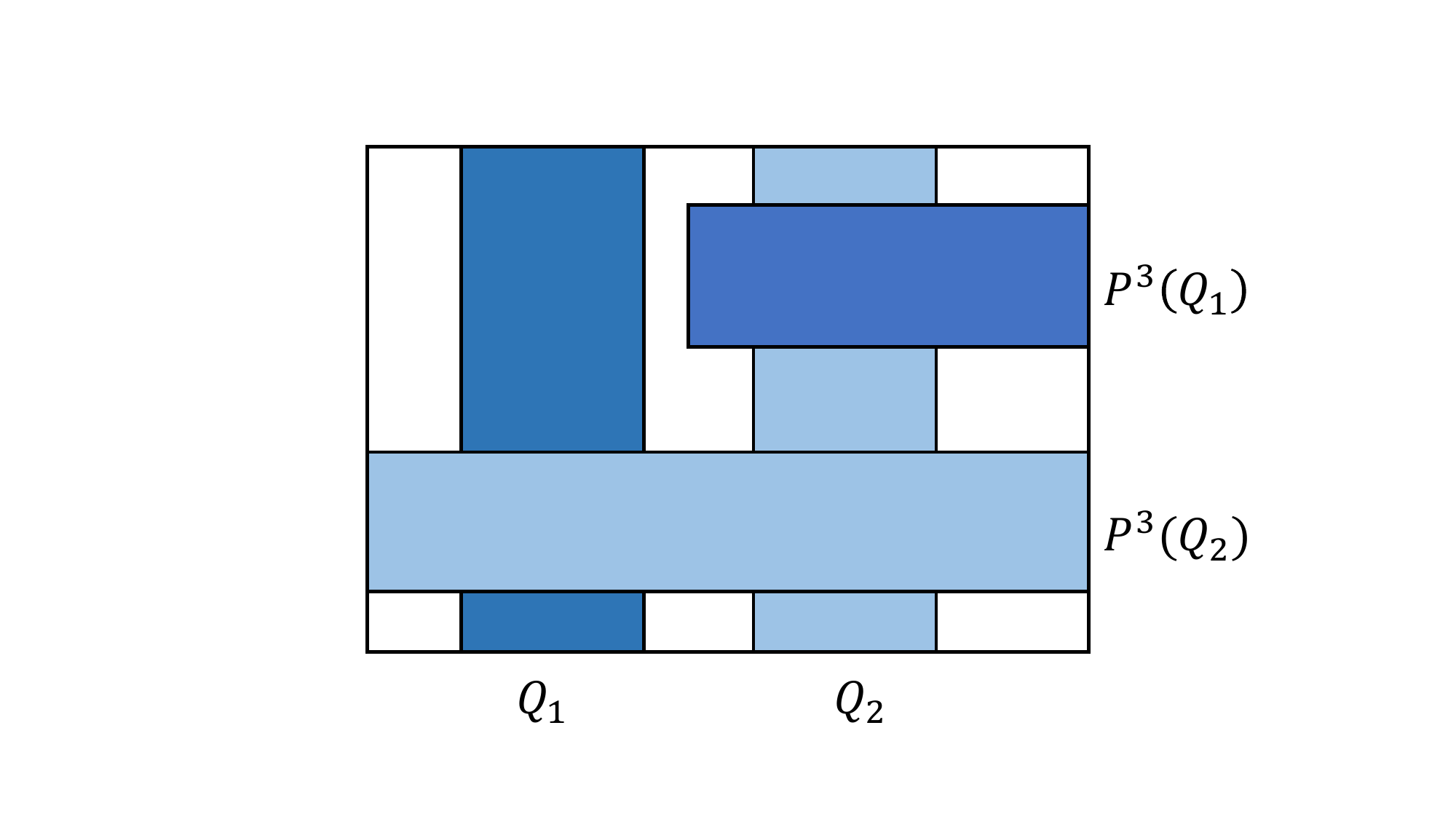}
	\caption{Schematic diagram of the pseudo-horseshoe in Figure~\ref{fig10:a}.}
	\label{fig10:b}
\end{figure}

Through the iteration of the map, we can investigate cases where the Smale horseshoe is absent. Notably, when $\gamma=1.184$, our analysis on the pseudo-horseshoe yields \( h(P_{\gamma=1.184}) \geq \log\frac{1+\sqrt{5}}{2} \approx 0.69424 \), whereas the Smale horseshoe of the second return map indicates \( h(P_{\gamma=1.184}) \geq \frac{1}{2} \log 2 = 0.5 \). This demonstrates that the first return map provides a more accurate lower bound than the second return map does, implying that iteration may lead to a loss of information. Although this issue could be addressed through additional iterations as what we do in Theorem \ref{SHT}, such methods are impractical for numerical experiments.

\section{Attractivity of the chaotic invariant set in the perturbed Duffing system}\label{experiment2}

The result obtained via the Melnikov method, together with the argument developed in Sec.~\ref{experiment1}, provides evidence for the existence of a chaotic invariant set for force amplitudes \(\gamma \in [0.4,1.1929]\). 
However, the question of whether this invariant set is attractive remains open. As remarked in Ref.~\cite{Chen2013}, the basin of attraction of a chaotic invariant set may have Lebesgue measure zero.

In this section, we investigate the asymptotic behavior of the perturbed Duffing system by means of two complementary approaches. The first is the computation of Lyapunov exponents; the second is an examination of whether the orbit converges to a fixed point, inspired by Cauchy's convergence criterion.

More precisely, we assess convergence by monitoring the successive distances
\[
\rho_n=\|(x_{n+1},y_{n+1})-(x_n,y_n)\|.
\]
Fix an integer \(m\), a small threshold \(\varepsilon>0\), and a sufficiently large integer \(N\). If
\[
\rho_{n}\le \varepsilon, \qquad \text{for all } n=N+1,\dots, N+m,
\]
then we deduce that the orbit converges to a fixed point.

We first compute the Lyapunov exponents for orbits with the fixed initial condition \((0.1,0.3)\)  by varying \(\gamma\in[0.4,0.5]\). The results are displayed in Figure~\ref{threshold}.
\begin{figure}[H]
	\centering	
	\includegraphics[height=6cm]{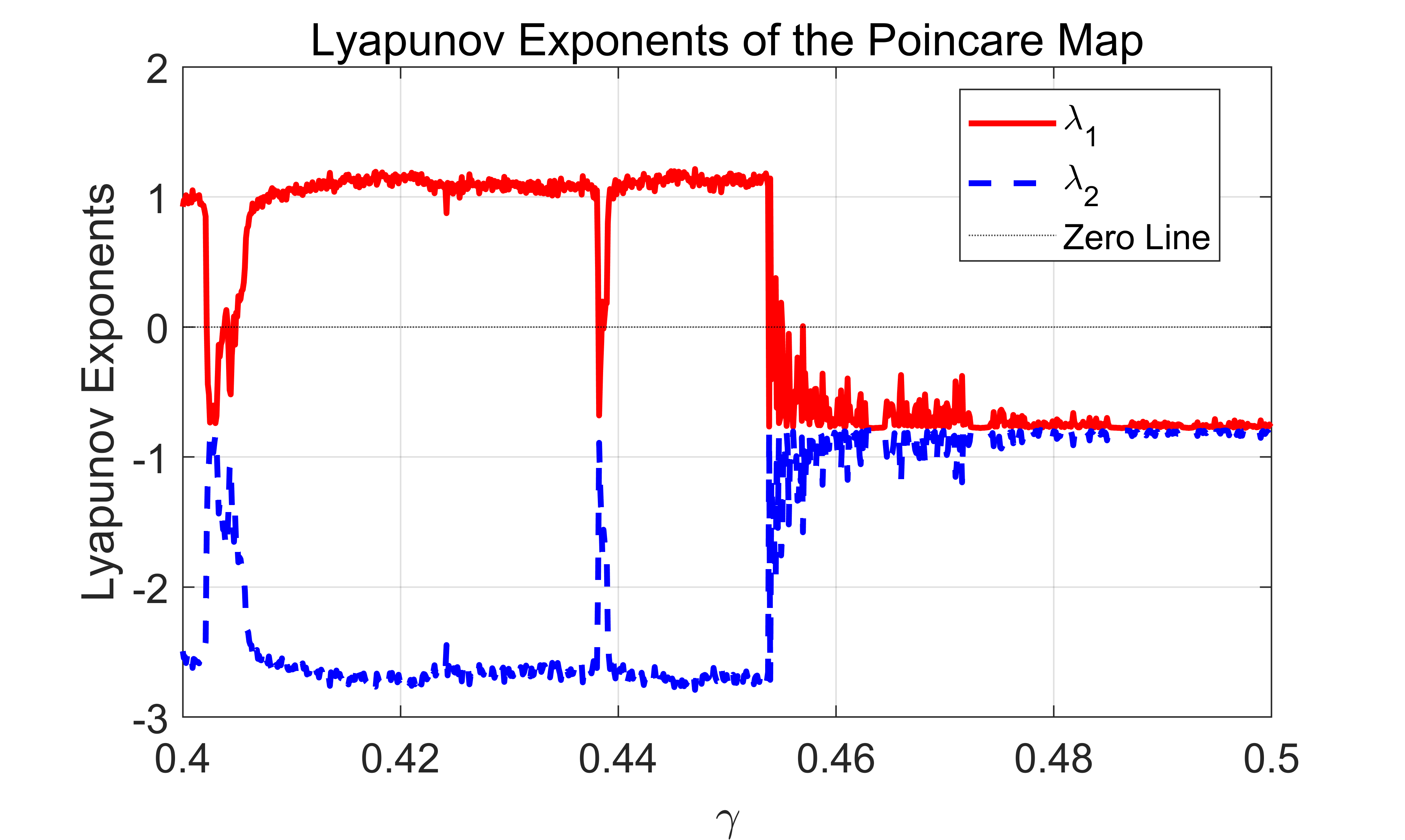}
	\caption{Lyapunov exponents vs. external force amplitude.}
	\label{threshold}
\end{figure}
It is evident that there exists a threshold value \(\gamma_0\in(0.45,0.46)\) such that both Lyapunov exponents become negative for all \(\gamma>\gamma_0\).
Further numerical evidence suggests that almost all the orbits converges to a fixed point, corresponding to a harmonic solution of \eqref{Duffing3}.

Our results further demonstrate that the attractivity of the chaotic invariant set does not vanish abruptly. Rather, the loss of attractivity can be observed through the progressive narrowing of its domain of attraction.
Fixing $N=100$, $m=2$, $\varepsilon=10^{-4}$, we classify points in the region $[-4,4]\times[-5.5, 6]$ based on whether the orbit starting from each such point converges to a fixed point.
For $\gamma\in\{0.42,0.43,0.44,0.45\}$, we apply the same classification procedure and obtain the results shown in Figure~\ref{fig12:a}--\ref{fig12:d}, where  the orbits starting from the blue region will converge to the chaotic invariant set, while those starting from the red region will converge to the fixed point. 

\begin{figure}[H]
	\centering
	\includegraphics[width=\textwidth]{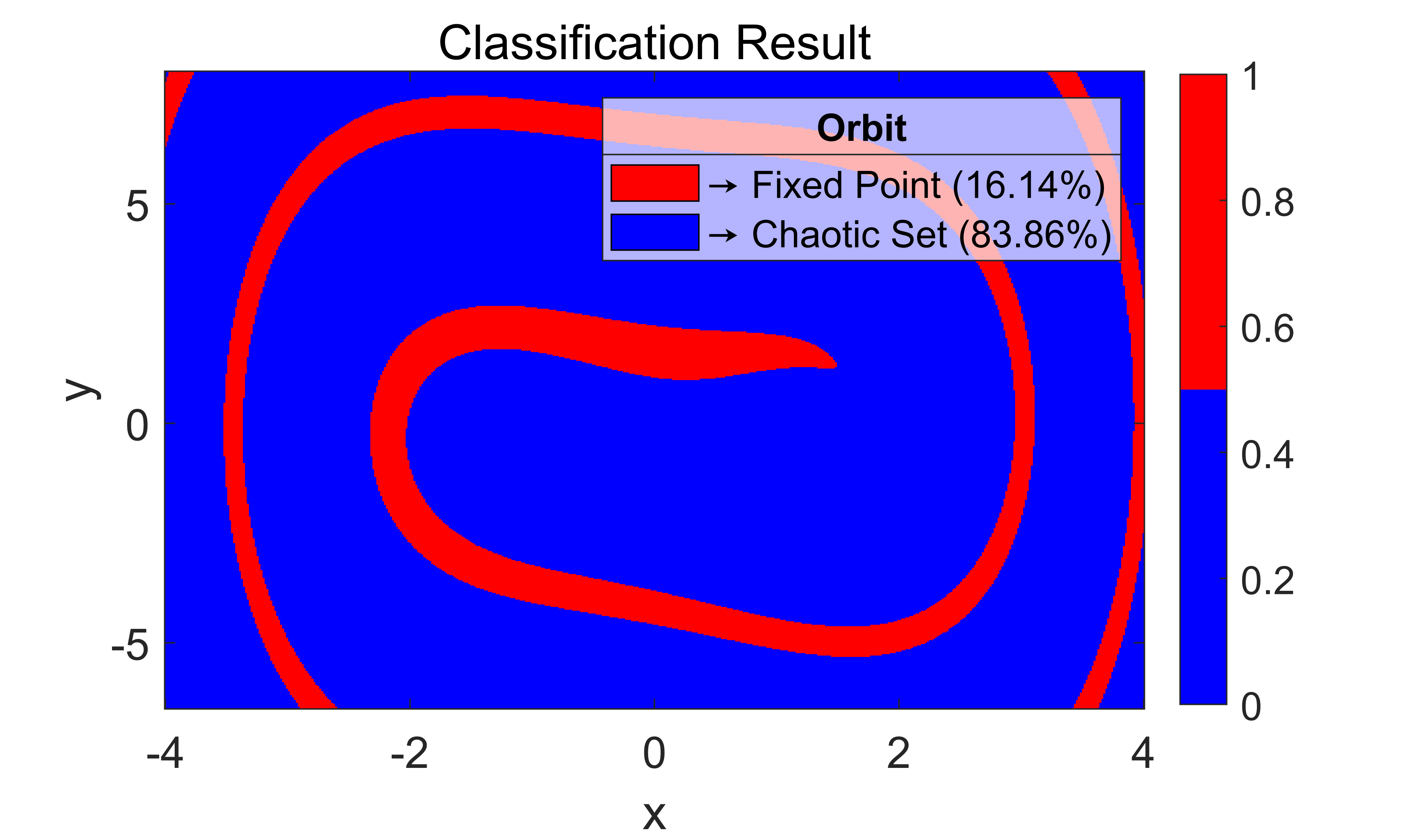}
	\caption{$\gamma=0.42$}
	\label{fig12:a}
\end{figure}
\begin{figure}[H]
	\centering
	\includegraphics[width=\textwidth]{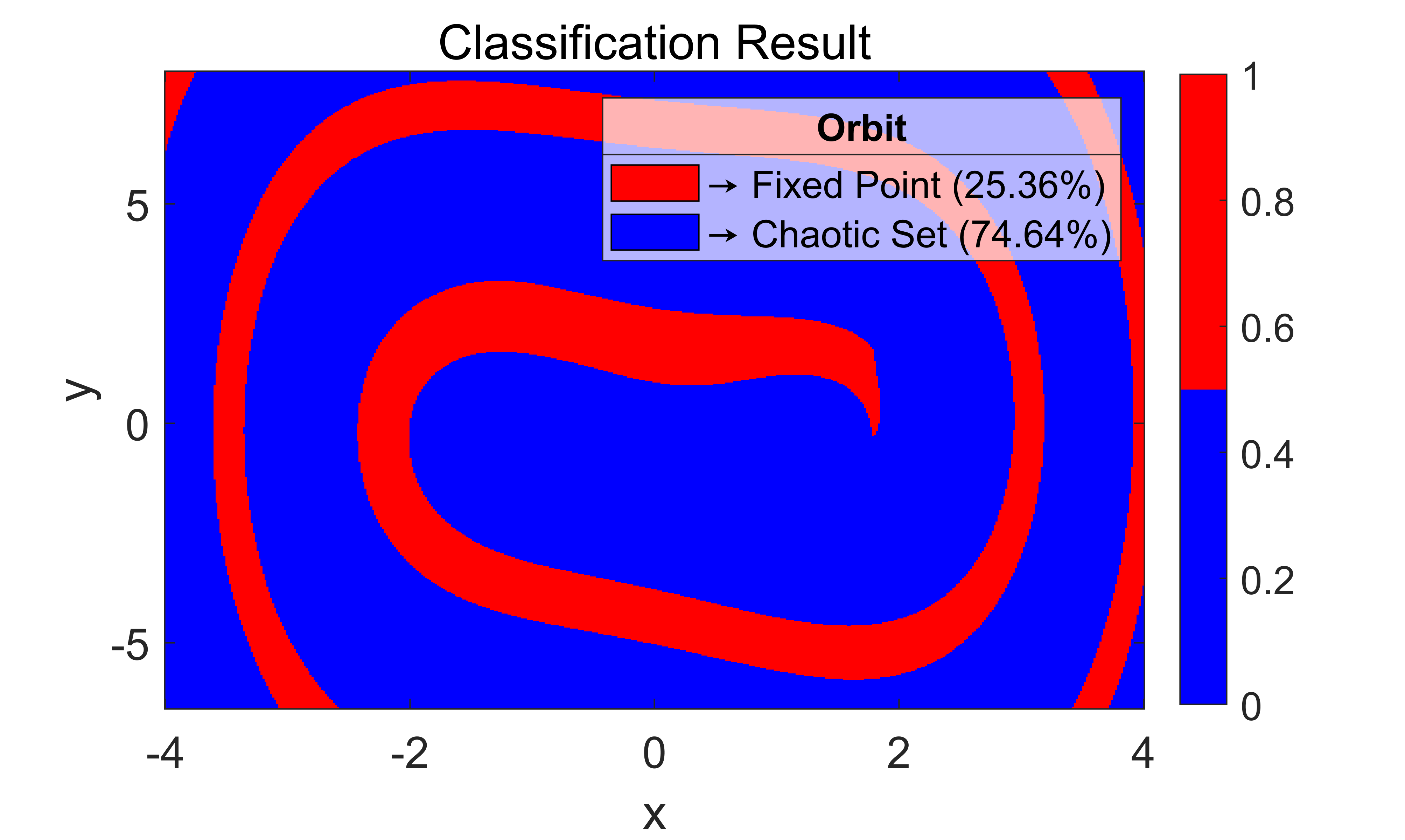}
	\caption{$\gamma=0.43$}
	\label{fig12:b}
\end{figure}
\begin{figure}[H]
	\centering
	\includegraphics[width=\textwidth]{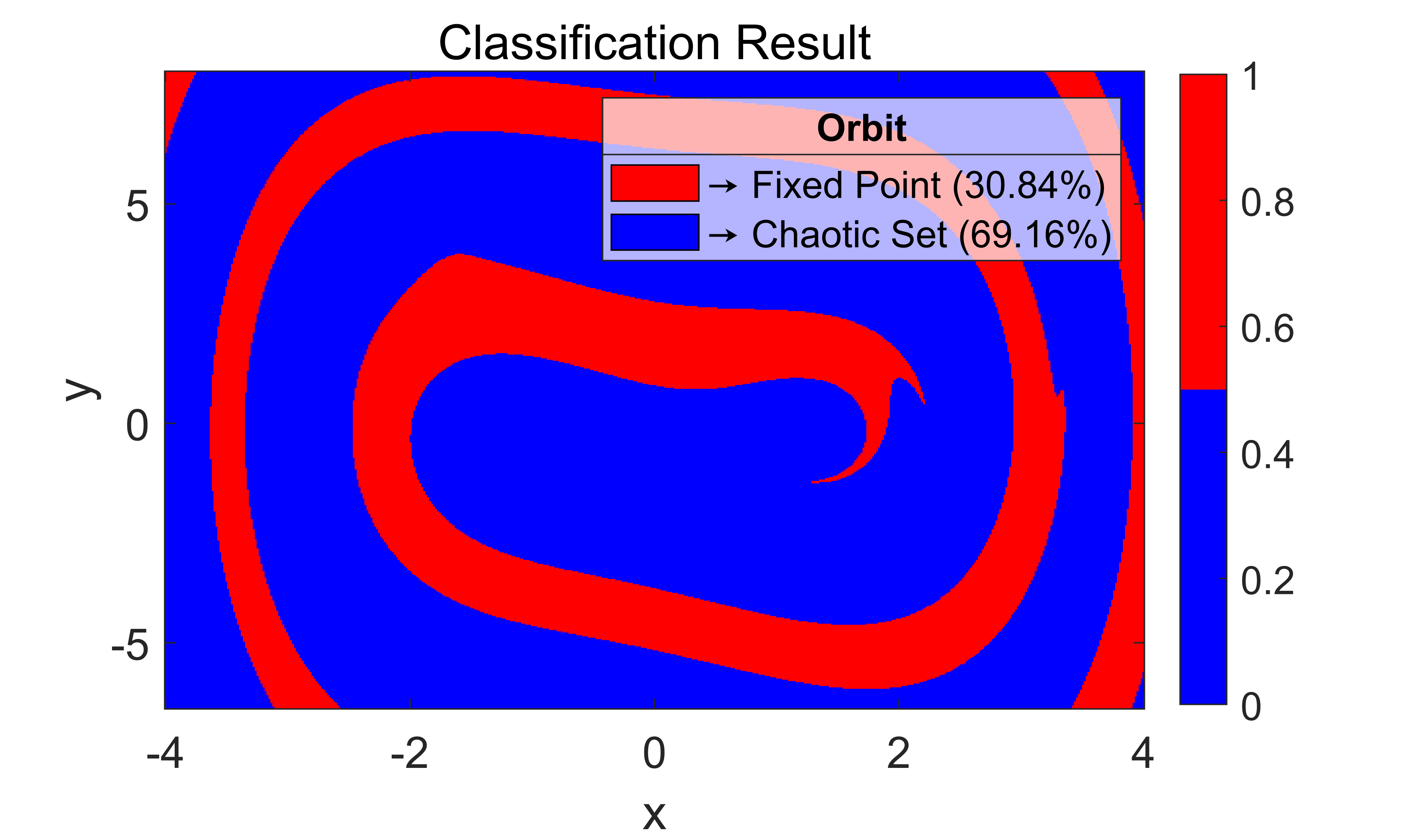}
	\caption{$\gamma=0.44$}
	\label{fig12:c}
\end{figure}
\begin{figure}[H]
	\centering
	\includegraphics[width=\textwidth]{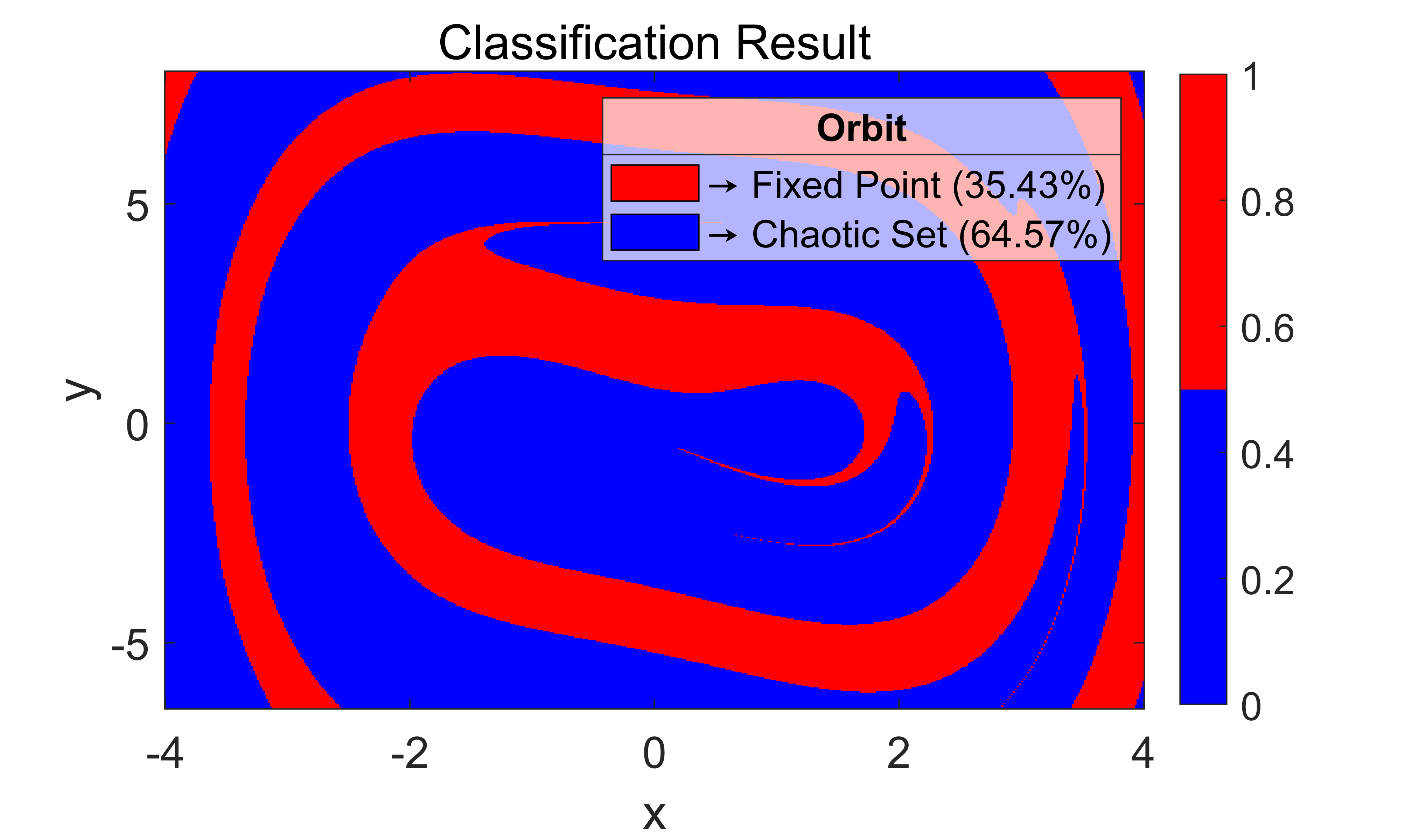}
	\caption{$\gamma=0.45$}
	\label{fig12:d}
\end{figure}

By computing the area ratio of the blue and red regions, we observe that the attraction basin of the chaotic invariant set shrinks as the amplitude $\gamma$ increases.
Consequently, the initial point $(0.1,0.3)$ lies in the basin of attraction of the fixed point for all $\gamma>\gamma_0$.

\begin{remark}
	It remains an open question whether the Lebesgue measure of the basin of attraction of the chaotic invariant set tends to zero as the amplitude increases. Our discussion provides only numerical evidence while a	comprehensive mathematical analysis is required to confirm our numerical findings.
\end{remark}

\begin{remark}
	It is noteworthy that for $\gamma$ near $\gamma_1=0.402~703$ and $\gamma_2=0.438~438$, both exponents are negative. This phenomenon may be attributed to the intricate geometric structure of the attraction basin associated with the chaotic invariant set, which requires further study.
	
\end{remark}

\section{Conclusion}\label{conclusion}

Using the topological horseshoe and Lyapunov exponents, we observe two main changes in  chaotic behavior of the perturbed Duffing system as the amplitude of the force amplitude increases. 
To be specific, the decrease in the lower bound of the topological entropy indicates that the system's chaotic dynamics gradually weaken as the amplitude $\gamma$ increases from $0.4$. 
In addition, the Lyapunov exponents reveal the existence of a threshold value \(\gamma_0 \in (0.45, 0.46)\) above which (and at least up to \(\gamma=1.1921\) in this paper) the chaotic invariant set loses its attractivity, even though the associated Smale horseshoe persists.
Furthermore, our simple classification procedure suggests that the basin of attraction of the chaotic invariant set shrinks as the amplitude \(\gamma\) increases from \(0.42\) to \(0.45\).

\bibliography{TH}

\end{document}